%% file: ho-ht.tex
\documentclass[online]{tlp}
\usepackage{fixkeywords}
\usepackage[arxiv]{preamble}
\usepackage{amsfonts}
\usepackage{amssymb}
\usepackage{amsmath}
\usepackage{latexsym}
\usepackage{stmaryrd}
\usepackage{microtype}

\ifarxiv
\usepackage{hyperref}
\fi

\usepackage[final]{listings}
\usepackage{caption}
\usepackage{booktabs}
\usepackage{etoolbox}
\usepackage{xspace}
\usepackage{customcommands}

\AtEndEnvironment{example}{\null\hfill$\Box$}

\begin{document}
\input{titleinfo}

\maketitle

\begin{abstract}
One of the most significant achievements of equilibrium logic was the
characterization of strong equivalence, a property crucial for program
transformation and optimization in Answer Set Programming (ASP). While ASP has
recently been extended to a higher-order setting to enhance its expressive
power, the lack of a comparable purely logical foundation has made verifying
strong equivalence for higher-order programs or even proving the correctness of
simple program transformations, a difficult challenge. This paper addresses this
gap by developing a logical semantics for higher-order ASP by extending the
equilibrium logic framework.
Within this extended framework we demonstrate that every stratified higher-order logic
program possesses a unique equilibrium model. Moreover, we establish  definability results
demonstrating that the syntax of our higher-order language is sufficiently expressive to
capture its semantic domains. Finally, and most importantly, we generalize the classical
theorem of strong equivalence to the higher-order setting: we prove that two
programs are strongly equivalent if and only if they share the same higher-order
models.
\ifarxiv
Under consideration for publication in Theory and Practice of Logic Programming (TPLP).
\fi
\end{abstract}
\section{Introduction}
%
%

A turning-point for answer set programming (ASP) was the purely logical
characterization of its semantics by~\cite{Pearce96,Pearce1999} based on the logic of
here-and-there (HT). Pearce developed a nonmonotonic system, called
\emph{equilibrium logic}, which can be described as a form of minimal model
reasoning in HT, and demonstrated that the equilibrium models of any program
coincide with its  answer sets. One of the most seminal successes of equilibrium
logic was the logical characterization of the notion of
\emph{strong equivalence} of logic programs by~\cite{LifschitzPV01},
an important and useful concept which has received considerable
research attention in the ASP community.

Recently, a higher-order extension of ASP was introduced \citep{iclp24} using the
abstract framework of Approximation Fixpoint Theory (AFT)~\citep{DMT00ApproximationsStableOperatorsWell-FoundedFixpointsApplications,DMT04Ultimateapproximationapplicationnonmonotonicknowledgerepresentation}. The
introduction of higher-order ASP is driven by more than just theoretical
curiosity: recently, many non-trivial extensions of ASP have been investigated
in order to enhance its
power~\citep{BJT16Stable-unstablesemanticsBeyondNPnormallogic,ART19BeyondNPQuantifyingoverAnswerSets,FLRSS21PlanningIncompleteInformationQuantifiedAnswerSet},
and thus it was a natural step to investigate the potential of developing a
higher-order extension of the paradigm. Not surprisingly, higher-order ASP
demonstrates impressive expressive power: as shown in the work of~\cite{iclp24} several difficult
computational problems were modeled in higher-order ASP, such as for example the
Generalized Geography two-player game, which is a
well-known~\citep{LS80GOPolynomial-SpaceHard} $\mathsf{PSPACE}$-complete
problem. The increased expressive power of higher-order ASP was theoretically
confirmed very recently: as demonstrated by~\cite{iclp25}, $(k+1)$-order ASP
captures \coNEXPTIME[k] using cautious reasoning and \NEXPTIME[k] using brave
reasoning. So, even the second-order fragment is extremely powerful and can
express problems that are beyond the capabilities of standard ASP systems.
Although expressive and quite promising, the idea of higher-order ASP lacks in
one important respect: due to the fact that its current characterization
by~\cite{iclp24} is based on a fixpoint construction, proving the correctness of
even simple transformations involves laborious and intricate double inductive
arguments (see for example~\citep[Appendix C]{iclp25}). In particular, defining
the strong equivalence of higher-order ASP programs based on the fixpoint
semantics from~\cite{iclp24}, appears to be a formidable challenge.

Based on the above discussion, a natural next step is to seek a logical
characterization of higher-order ASP, which would allow notions such as strong
equivalence to be defined in a simple way and reasoning about programs to be
performed in a purely logical manner. Such a characterization requires extending
HT to higher-order types and accordingly generalizing the framework of
equilibrium logic. But what form should these higher-order HT relations have?
Should they be arbitrary, or should they obey specific restrictions to capture
the computational nature of logic programs? Drawing on well-known concepts from
denotational semantics~\citep{tennent,gunter}, we adopt the view that meaningful
higher-order relations in a computational context should be monotonic over
appropriate semantic domains. We introduce \emph{justification monotonic}
relations which prove to be a cornerstone of our work, essential for
generalizing the characterization of strong equivalence to higher-order
programs. The paper's main contributions are outlined below:
\begin{itemize}
\item We develop a purely logical semantics for the higher-order logic
      programming language $\HOL$. Our approach extends the 3-valued truth
      domain of HT to the full type hierarchy of $\HOL$, restricting the domain
      to \emph{justification-monotonic} relations at each level of the
      hierarchy. Within this framework, we define \emph{equilibrium models} for
      $\HOL$ as the \emph{total} interpretations
      that are minimal with respect to the justification ordering. We argue that
      this restriction is necessary: adhering to a standard model of types -
      where the arrow type constructor denotes the set of all 3-valued relations
      - yields counterintuitive results even for trivial programs. Finally, we
      demonstrate that under this proposed semantics, \emph{stratified} $\HOL$
      programs always possess a unique equilibrium model.

\item We examine \emph{definability} concepts regarding the proposed semantics.
      More specifically, we demonstrate that every \emph{total}
      justification-monotonic relation can be represented by a stratified $\HOL$
      program. Moreover, every \emph{non-total} justification-monotonic relation
      can be ``captured'' (in a sense that will be precisely defined in
      Section~\ref{sec:representation-theorem}) by a stratified $\HOL$ program.
      Such definability results have a foundational significance: they
      demonstrate that the syntactic elements of $\HOL$ are strong enough to
      express all the elements of its semantic domains.       To use Robin
      Milner's words~\citep{Milner77}, definability results suggest that the
      proposed semantics is not ``over-generous''.

\item We generalize the classical proof of propositional strong equivalence to
      the world of higher-order logic programs. More specifically, we
      demonstrate that two $\HOL$ programs are strongly equivalent iff they have
      the same $\HOL$ models. Our proof relies significantly on the
      justification-monotonic nature of the proposed semantics and the derived
      definability results for $\HOL$. This result paves the road for further
      investigations on transformations and verification of higher-order logic
      programs, an area that at present is largely unexplored.
\end{itemize}
We believe that the present work can be the starting point for further
investigations not only in the logical foundations of higher-order logic
programs but also in their practical applications, extending in this way
significantly the scope of the ASP paradigm.

\section{$\HOL$: A Higher-Order Logic Programming Language}\label{sec:hol}
In this section we define the syntax of the language $\HOL$ that we use
throughout the paper. Actually, $\HOL$ is based on a simple type system with one
\emph{base} type, namely $\bool$, the Boolean domain.
%
\begin{definition}
The types of $\HOL$ are denoted by $\pi$ (and its subscripted
versions), and are defined as:
\begin{align*}
\pi  & := \bool \mid (\pi_1 \to \pi_2)
\end{align*}
\end{definition}

As usual, the binary operator $\to$ is right-associative. It can be easily seen
that every predicate type $\pi$ can be written in the form
$\pi_1 \to \cdots \to \pi_n \rightarrow \bool$, $n\geq 0$ (for $n=0$
we assume that $\pi=\bool$).
We proceed by defining the syntax of $\HOL$.

\begin{definition}
The \emph{alphabet} of $\HOL$ consists of the following: \emph{predicate variables}
of every type $\pi$ (denoted by capital letters such as $\mathsf{P,Q,R,\ldots})$;
\emph{predicate constants} of every type $\pi$ (denoted by lowercase letters such as $\mathsf{p,q,r,\ldots}$);
the \emph{conjunction} constant $\wedge$ of type $\bool \to \bool \to \bool$;
and the \emph{negation} constant $\pnot$ of type $\bool \to \bool$.
%
\end{definition}
\begin{definition}\label{def:expressions}
The \emph{expressions} and \emph{literals} of $\HOL$ are defined as follows.
Every predicate variable or constant is an expression; if $\mathsf{E}_1$ is an expression of type $\pi_1 \to \pi_2$ and $\mathsf{E}_2$ an
expression of type $\pi_1$ then $(\mathsf{E}_1\ \mathsf{E}_2)$ is an expression of type $\pi_2$.
If $\mathsf{E}$ is an expression of type $\bool$ then $\mathsf{E}$ is also a (positive) \emph{literal} of type $\bool$;
if $\mathsf{E}$ is an expression of type $\bool$ then $(\pnot \mathsf{E})$ is a (negative) \emph{literal} of type $\bool$.
\end{definition}
We will omit parentheses when no confusion arises and in some other cases we may use additional ones if we feel that
this enhances readability; so, for example, when a predicate constant \texttt{p} is unary, we may write \mbox{\texttt{p(R)}}
instead of \mbox{\texttt{p R}}. Literals will be denoted by $\mathsf{L}$ and its subscripted versions.
We write $\mathsf{E}:\pi$ to denote that an expression $\mathsf{E}$ has type $\pi$.
\begin{definition}
A \emph{rule} of $\HOL$ is a syntactic construct of the form
$\mathsf{p}\ \mathsf{R}_1 \cdots \mathsf{R}_n \lrule \mathsf{L}_1 \land \cdots \land \mathsf{L}_m$,
where $\mathsf{p}$ is a predicate constant of type $\pi_1 \to \cdots \to \pi_n \to \bool$,
$\mathsf{R}_1,\ldots,\mathsf{R}_n$ are distinct variables of types $\pi_1,\ldots,\pi_n$ respectively and
the $\mathsf{L}_i$ are literals.
The term $\mathsf{p}\ \mathsf{R}_1 \cdots \mathsf{R}_n$ is the \emph{head} of the rule and
$ \mathsf{L}_1 \land \cdots \land \mathsf{L}_m$ is its \emph{body}.
A \emph{program} $\mathsf{P}$ of $\HOL$ is a finite set of rules.
\end{definition}

We will often follow the common logic programming notation and write
$\mathsf{L}_1,\ldots,\mathsf{L}_m$ instead of
$\mathsf{L}_1 \wedge \cdots \wedge \mathsf{L}_m$ for the body of a rule.
For brevity reasons, we will often denote
a rule as $\mathsf{p} \ \overline{\mathsf{R}} \lrule \mathsf{B}$, where
$\overline{\mathsf{R}}$ is a shorthand for a sequence of variables
$\mathsf{R}_1 \cdots \mathsf{R}_n$ and $\mathsf{B}$ represents a
conjunction of literals.

\begin{example}\label{syntax-example}
The following predicates illustrate the syntax of $\HOL$. Their precise semantics, particularly the distinction between $\mathtt{id}$ and $\mathtt{neg\_neg}$, will be clarified in the next section.
\[
\begin{array}{lcl}
  \mbox{\texttt{id(R)}} & \leftarrow & \mathtt{R}\\
  \mbox{\texttt{neg(R)}} & \leftarrow & \pnot\mathtt{R}\\
  \mbox{\texttt{neg\_neg(R)}} & \leftarrow & \mbox{\texttt{neg(neg(R))}}\\
  \mathtt{eq}~\mathtt{R}~\mathtt{Q} & \leftarrow & \mbox{\texttt{neg(R),neg(Q)}}\\
  \mathtt{eq}~\mathtt{R}~\mathtt{Q} & \leftarrow & \mbox{\texttt{neg\_neg(R),neg\_neg(Q)}}
\end{array}
\]
The predicate constants \texttt{id}, \texttt{neg}, and \texttt{neg\_neg} are of type $o\to o$, while \texttt{eq}
is of type $o \to (o \to o)$. The constant \texttt{id} denotes the identity relation, while \texttt{neg} lifts
the negation operator $\pnot$ to a predicate constant; this definition allows us to compose negations in rule bodies,
as seen in the definition of \texttt{neg\_neg}; this achieves the effect of double negation (i.e., $\pnot\pnot\mathtt{R}$),
which is not directly supported as a syntactic form in $\HOL$. Finally, \texttt{eq} defines an equality relation
on truth values, whose semantics will be explained in Section~\ref{sec:semantics_of_HOL}.
\end{example}

\section{Equilibrium Semantics for $\HOL$ Programs}\label{sec:semantics_of_HOL}

In this section we introduce the semantics of $\HOL$ programs, which is based on an extension of the equilibrium
semantics of propositional programs introduced by~\cite{Pearce96,Pearce1999}. Pearce's equilibrium semantics is built on
top of Heyting's logic of here-and-there (HT); the semantics of HT is usually presented either through Kripke frames
or through 3-valued truth tables. We find it more convenient to use the latter approach, however we believe that our definitions and proofs
can be adapted to the former. The underlying set of truth values of HT is $V = \{\F, \NF, \T\}$. While $\F$ and $\T$ correspond
to false and true, the value $\NF$ is an intermediate one between $\F$ and $\T$, whose intuitive reading is
``\emph{weakly true because it lacks sufficient justification to be considered true}''. The standard ordering
in $V$ is $\F < \NF < \T$. Consider now the following ordering over $V$, which we represent by $\preceq$ and
call the \emph{justification ordering}:
\begin{definition}
For all $v_1,v_2 \in \{\F, \NF, \T\}$, $v_1 \preceq v_2$ iff $v_1=v_2$ or $v_1=\NF$ and $v_2=\T$.
\end{definition}
The justification ordering is not new: it is the standard ordering used by D. Pearce over models in order to define the equilibrium ones.
We use this ordering here to restrict our semantic domains to contain relations that preserve this ordering, namely the \emph{justification monotonic} ones. The intuitive reason behind this requirement is that if a (higher-order) relation yields a $\NF$ value, it must maintain or increase that value when the justification of its input increases. In other words, we don't want stronger evidence to invalidate a previous (even weakly true) conclusion. Later in this section, we provide a simple example which additionally motivates our adoption of justification monotonicity.


We proceed with the formal semantics of $\HOL$, starting with the semantics of its types.
The semantics of the base Boolean domain is the set $V = \{\F, \NF, \T\}$. The semantics of types of
the form $\pi_1\rightarrow \pi_2$ is the set of {\em justification-monotonic} functions
from the domain of type $\pi_1$ to that of type $\pi_2$. We define, simultaneously
with the meaning of every type $\pi$, two partial orders on the elements of type $\pi$:
the relation $\leq_{\pi}$ which represents the {\em truth} ordering, and the relation
$\preceq_{\pi}$ which represents the {\em justification} ordering.
\begin{definition}\label{our_domains}
Let $\mathsf{P}$ be a $\HOL$ program. For every type $\pi$ we define its meaning $\lsem \pi \rsem$, as follows:

  \begin{itemize}
    \item $\lsem o \rsem = \{ \F, \NF, \T \}$. The partial order $\leq_o$
          is the usual one induced by the ordering  $\F <_o \NF <_o \T$;
          the partial order $\preceq_o$ is the one induced by the ordering  $\NF \prec_o \T$.

    \item $\lsem \pi_1 \rightarrow \pi_2 \rsem = [ \lsem \pi_1 \rsem \rightarrow
          \lsem \pi_2 \rsem ]$, namely the $\preceq$-monotonic
          functions\footnote{Function $f\in \lsem \pi_1 \rightarrow \pi_2\rsem$ is $\preceq$-monotonic if for all
          $d_1,d_2\in \lsem \pi_1\rsem$, $d_1\preceq_{\pi_1} d_2$ implies $f(d_1)\preceq_{\pi_2}f(d_2)$.} from \lsem $\pi_1 \rsem$ to $\lsem \pi_2 \rsem$.
          The partial order $\leq_{\pi_1 \rightarrow \pi_2}$ is defined as follows: for all $f,g \in \lsem \pi_1 \rightarrow \pi_2 \rsem$,
          $f \leq_{\pi_1\rightarrow\pi_2} g$ iff $f(d) \leq_{\pi_2} g(d)$ for all $d \in \lsem \pi_1 \rsem$. The partial order
          $\preceq_{\pi_1 \rightarrow \pi_2}$ is defined as follows: for all $f,g \in \lsem \pi_1 \rightarrow \pi_2 \rsem$,
          $f \preceq_{\pi_1\rightarrow\pi_2} g$ iff $f(d) \preceq_{\pi_2} g(d)$ for all $d \in \lsem \pi_1 \rsem$.
  \end{itemize}
\end{definition}
We will omit the subscripts of the above orders when they are obvious from context.

\begin{definition}\label{def:interpretation}
An \emph{interpretation} ${\cal I}$ assigns to each predicate constant $\mathsf{p} : \pi$ of $\HOL$, an element
${\cal I}(\mathsf{p}) \in \mo{\pi}$.
\end{definition}

We will denote the set of interpretations with ${\cal H}$. We define two partial orders on ${\cal H}$ as
follows: for all ${\cal I}, {\cal J} \in {\cal H}$, ${\cal I} \leq {\cal J}$
(respectively, ${\cal I} \preceq {\cal J}$) iff for every predicate constant
$\mathsf{p}$, ${\cal I}(\mathsf{p}) \leq_{\pi} {\cal J}(\mathsf{p})$
(respectively, ${\cal I}(\mathsf{p}) \preceq_{\pi} {\cal J}(\mathsf{p})$).
\begin{definition}\label{def:state}
	A \emph{state} $s$ is a function that
	assigns to each predicate variable $\mathsf{R}:\pi$, an element
	$s(\mathsf{R}) \in \mo{\pi}$.
\end{definition}

We denote the set of states with ${\cal S}$. The partial orders
$\leq$ and $\preceq$ extend to ${\cal S}$ in the obvious way.
We will often use $s[\mathsf{R}_1/d_1,\ldots,\mathsf{R}_n/d_n]$ to denote a state that is identical to $s$ the
only difference being that the new state assigns to each $\mathsf{R}_i$ the corresponding value $d_i$;
for brevity, we will also denote it by $s[\overline{\mathsf{R}}/\overline{d}]$.

By abuse of language, we will often talk about ``an interpretation (respectively, state) of
a given program $\mathsf{P}$'' instead of ``an interpretation (respectively, state) of $\HOL$''.

\begin{definition}
Let $\mathcal{I}$ be an interpretation and $s$ a state. The semantics of expressions and literals with respect to ${\cal I}$ and $s$
is defined as follows:
\begin{enumerate}
  \item $\mwrs{\mathsf{R}}{\mathcal{I}}{s} = s(\mathsf{R})$
  \item $\mwrs{\mathsf{p}}{\mathcal{I}}{s} = \mathcal{I}(\mathsf{p})$
  \item $\mwrs{(\mathsf{E}_1\ \mathsf{E}_2)}{\mathcal{I}}{s} = \mwrs{\mathsf{E}_1}{\mathcal{I}}{s}(\mwrs{\mathsf{E}_2}{\mathcal{I}}{s})$
  \item $\mwrs{\mathsf{(\pnot E)}}{\mathcal{I}}{s}= \pnot \mwrs{\mathsf{E}}{\mathcal{I}}{s}$, where $\pnot \T = \F$, $\pnot \F = \T$, and $\pnot \NF = \F$.
\end{enumerate}
\end{definition}
We can now formally define the notion of \emph{model} for $\HOL$ programs.
\begin{definition}
Let $\mathsf{P}$ be a program and $\mathcal{M}$ be an interpretation of $\mathsf{P}$.
Then, $\mathcal{M}$ is a {\em model of $\mathsf{P}$} iff for every rule
$\mathsf{p}\ \mathsf{V}_1\cdots\mathsf{V}_n \leftarrow \mathsf{L}_1, \ldots, \mathsf{L}_m$ in $\mathsf{P}$
and for every state $s\in {\cal S}$, it holds that
$\min_{\leq}\{\mwrs{\mathsf{L}_i}{\mathcal{M}}{s}\mid i\in\{1,\ldots,m\}\}\leq\mwrs{\mathsf{p}\ \mathsf{V}_1\cdots\mathsf{V}_n}{\mathcal{M}}{s}$.
\end{definition}

We will be particularly interested in semantic elements that are essentially 2-valued:
\begin{definition}
Let $\pi$ be a type. An element $e \in \mo{\pi}$ is called  \emph{total} iff:
\begin{enumerate}
  \item $\pi = o$ and $e \in \{\F,\T\}$, or
  \item $\pi = \pi_1\to\pi_2$ and $e(e')$ is total for all total $e'\in \mo{\pi_1}$.
\end{enumerate}
An interpretation $\mathcal{I}$ will be called \emph{total} iff $\mathcal{I}(\mathsf{p})$ is total for every predicate constant $\mathsf{p}$.
\end{definition}

The semantics of a $\HOL$ program will be captured by its set of \emph{equilibrium} models:
\begin{definition}
Let $\mathsf{P}$ be a $\HOL$ program and $\mathcal{M}$ be an interpretation of $\mathsf{P}$.
Then, $\mathcal{M}$ is called an \emph{equilibrium model of $\mathsf{P}$} iff
$\mathcal{M}$ is total and a $\preceq$-minimal model of $\mathsf{P}$.
\end{definition}

Therefore, the meaning of a $\HOL$ program is captured by its set of equilibrium models. This semantics
generalizes the classical equilibrium model semantics of Pearce for propositional programs~\citep{Pearce96,Pearce1999} (because
it uses the same logical machinery at the propositional level). In the rest of this section we give some examples
of the proposed semantics, illustrate the necessity of restricting attention to $\preceq$-monotonic relations, and
state two important properties of the semantics.
\begin{example}
Based on the above semantics, it is not hard to check that the program consisting of the definitions
in Example~\ref{syntax-example}, has a unique equilibrium model which assigns to \texttt{id} the relation
$\{(\F,\F),(\NF,\NF),(\T,\T)\}$, to \texttt{neg} the relation $\{(\F,\T),(\NF,\F),(\T,\F)\}$, and to
\texttt{neg\_neg} the relation $\{(\F,\F),(\NF,\T),(\T,\T)\}$. Notice the difference between \texttt{id}
and \texttt{neg\_neg} which demonstrates the fact that in HT double negation is different from the identity relation.
Finally, the meaning of \texttt{eq} is a relation that returns $\T$ either if its two arguments are the same or one of
them is $\NF$ and the other $\T$; otherwise, it returns $\F$. It is easy to verify that all the aforementioned
relations are $\preceq$-monotonic. On the other hand, the standard equality relation (which returns $\T$ only
when its two arguments are the same and $\F$ otherwise), is not $\preceq$-monotonic and therefore not definable
as a $\HOL$ program. In other words, $\HOL$ programs have no way to distinguish between $\NF$ and $\T$.
\end{example}
\begin{example}\label{two-valued-monotone}
Consider the following program:
\[
\begin{array}{lcl}
  \mbox{\texttt{neg(R)}} & \leftarrow & \pnot\mathtt{R}\\
  \mbox{\texttt{p(R)}} & \leftarrow & \mbox{\texttt{neg(neg(p(R)))}} %
\end{array}
\]
It can be verified that the above program has 4 equilibrium models, namely all the $\preceq$-monotonic
relations of type $o \to o$ that are ``two-valued'' (i.e., they always return a value in $\{\F,\T\}$).
More specifically, the four equilibrium models assign to \texttt{p} the following relations:
${\cal M}_1(\texttt{p}) =\{(\F,\F),(\NF,\F),(\T,\F)\}$,
${\cal M}_2(\texttt{p}) =\{(\F,\T),(\NF,\F),(\T,\F)\}$,
${\cal M}_3(\texttt{p}) =\{(\F,\F),(\NF,\T),(\T,\T)\}$, and
${\cal M}_4(\texttt{p}) =\{(\F,\T),(\NF,\T),(\T,\T)\}$.
\end{example}

The previous example demonstrates the production of all ``two-valued'' $\preceq$-monotonic
relations of type $o \to o$. Can we produce \emph{all the total} $\preceq$-monotonic ones? This is
a more difficult task, requiring a more complicated program \ifincludeappendix(see Appendix~\ref{appendix-of-section-3})\else(see the supplementary material)\fi.

We now demonstrate, from a different angle, the importance of justification monotonicity in our semantics.
The following simple example illustrates the necessity of the concept.
\begin{example}\label{necessity-example}
Consider the following program $\mathsf{P}$, defining the predicate {\tt p} of type $o\to o$:
\[
\begin{array}{lll}
  \mbox{\texttt{p(R)}} & \leftarrow & \mathtt{fail}
\end{array}
\]
Intuitively, we expect that the meaning of {\tt p} is the relation ${\cal M}({\tt p})=\{(\F, \F),(\NF, \F), (\T, \F)\}$ (ie., the
relation that maps every truth value to $\F$). However, if our semantics allowed arbitrary relations to be the meanings of predicates,
the above program would have another equilibrium model, namely
$\mathcal{N}({\tt p}) = \{(\F, \F),(\NF, \NF), (\T, \F)\}$.
Obviously, ${\cal N}({\tt p})$ is not monotonic, since $\NF \preceq \T$ but ${\cal N}({\tt p})(\NF) = \NF \not\preceq \F = {\cal N}({\tt p})(\T)$.
However, ${\cal N}$ is total and $\preceq$-minimal model of $\mathsf{P}$, so we would have to
accept it as an equilibrium model of $\mathsf{P}$ if we had not restricted our
domains to contain only the justification monotonic relations.

One could argue that accepting both ${\cal M}$ and ${\cal N}$ as equilibrium models, would be strange but not catastrophic; after all, the meaning of \texttt{p} is identical in both models if we restrict attention to total inputs. However, assume we add two more rules to the program:
\[
\begin{array}{lll}
  \mbox{\texttt{result}} & \leftarrow & \mbox{\texttt{neg(neg(p(Q)))}}\\
  \mbox{\texttt{p(R)}}   & \leftarrow & \mathtt{fail}\\
  \mbox{\texttt{neg(R)}} & \leftarrow & \pnot\mathtt{R}
\end{array}
\]
The definition of \texttt{result} uses in its body the existential variable \texttt{Q} (so, intuitively, the body reads ``there exists a variable \texttt{Q} such that ...''). It is not hard to see that the above program has two equilibrium models ${\cal M}$ and ${\cal N}$ such that ${\cal M}(\texttt{result}) = \F$ and ${\cal N}(\texttt{result}) = \T$. Obviously, now ${\cal N}$
produces a rather counterintuitive result.
\end{example}

We now establish two key properties of the proposed equilibrium semantics.
Beyond their technical utility in proving subsequent results, these properties
demonstrate that our semantics adheres to fundamental principles of logic
programming. The first is a splitting lemma, analogous to the one introduced
in~\citep{Lifschitz1994}. The second is a guarantee that the proposed semantics assigns a
unique equilibrium model to every \emph{stratified} $\HOL$ program. \ifincludeappendix The proofs of the
results are given in the appendix of the paper.\else The proofs of the
results are given in the supplementary material.\fi
\begin{definition}\label{def_splitting}
Let $\mathsf{P}$ be a program and $U$ a set of predicate constants.
We say that $U$ is a splitting set of $\mathsf{P}$
if for every rule $C$ of $\mathsf{P}$,
if a predicate constant of $U$ appears in the head of $C$, then every predicate constant appearing in $C$ is included in $U$.
The set of rules $C\in \mathsf{P}$ such that all predicate constants appearing in $C$ are included in $U$
is called the \emph{bottom} of $\mathsf{P}$ relative to $U$ and denoted by $b_U(\mathsf{P})$.
The set $\mathsf{P} \setminus b_U(\mathsf{P})$ is called the \emph{top} of $\mathsf{P}$ relative to $U$.
\end{definition}

\begin{lemma}\label{splitting_submodel}
Let $\mathsf{P}$ be a program and $U$ be a splitting set of $\mathsf{P}$.
If $\mathcal{M}$ is an equilibrium model of $\mathsf{P}$,
then $\mathcal{M}$ restricted to $U$ is an equilibrium model of $b_U(\mathsf{P})$.
\end{lemma}

\emph{Stratified higher-order logic programs}~\citep{iclp24} is a broad class of programs for which, as we show,
the equilibrium semantics always produces a unique model.
\begin{definition}
A $\HOL$ program $\mathsf{P}$ is called \emph{stratified} if there is a function $S$ mapping
predicate constants to natural numbers, such that for each rule
$\mathsf{p} \ \overline{\mathsf{R}} \leftarrow \mathsf{L}_1 \wedge \cdots \wedge \mathsf{L}_m$
and any $i\in\{1,\ldots, m\}$:
\begin{itemize}
  \item $S(\mathsf{q})\leq S(\mathsf{p})$ for every predicate constant $\mathsf{q}$ occurring in $\mathsf{L}_i$.
  \item If $\mathsf{L}_i$ is of the form $(\sim \!\mathsf{E})$, then $S(\mathsf{q})< S(\mathsf{p})$ for each predicate constant $\mathsf{q}$ occurring~in~$\mathsf{E}$.
  \item For any subexpression of $\mathsf{L}_i$ of the form $(\mathsf{E}_1~\mathsf{E}_2)$, $S(\mathsf{q})< S(\mathsf{p})$ for every predicate constant $\mathsf{q}$ occurring in $\mathsf{E}_2$.
\end{itemize}
\end{definition}
%
%
\begin{theorem}\label{is_unique}
Let $\mathsf{P}$ be a stratified $\HOL$ program. Then, $\mathsf{P}$ has a unique equilibrium model.
\end{theorem}

Notice that Examples~\ref{syntax-example} and~\ref{necessity-example} are stratified programs and, obviously, have
a unique equilibrium model.

\section{Two Definability Theorems for $\HOL$}\label{sec:representation-theorem}
In this section we demonstrate that every total $\preceq$-monotonic relation is definable by a stratified $\HOL$ program.
Additionally, we argue that every (even non-total) $\preceq$-monotonic relation can be ``captured'' by a stratified
$\HOL$ program (in a sense that will become clear later in the section). These definability results are important
for two main reasons:
\begin{itemize}
\item Both results play an important technical role in establishing the strong equivalence theorem
      (Theorem~\ref{se-theorem} in the next section). The significance of the results in the proof,
      will be intuitively explained in the next section.

\item Such definability results have an additional foundational significance: they demonstrate that $\HOL$ is
      strong enough to express all the elements of its semantic domain. To use Robin Milner's words~\citep{Milner77},
      such results suggest that the proposed semantics is not ``over-generous''. According to Milner, a semantics is
      over-generous if there exist elements in the semantic domain of a language that are not definable
      by phrases (ie., syntactic elements) of the language. It is interesting to note that Milner's remarks
      about definability, were given in the context of discussing \emph{full-abstraction},
      a notion conceptually close to strong equivalence.
\end{itemize}
We now formally state the two definability results and illustrate them by corresponding examples. \ifincludeappendix The proofs
of the theorems are given in the appendix of the paper.\else The proofs
of the theorems are given in the supplementary material.\fi The two examples have been obtained by
simplifying the programs produced by the constructions in the proofs of the theorems.
\begin{theorem}\label{definability_of_exact}
For each type $\pi$ and every total $d\in\mo{\pi}$, there exists a stratified program $\mathsf{P}_d$ with a constant $\ce{d} : \pi$ such that
$\mathsf{P}_d$ has a unique equilibrium model $\mathcal{M}$ with $\mathcal{M}(\ce{d}) = d$.
\end{theorem}
\begin{example}
Let $d \in \mo{o \to o}$ be a total relation such that $d=\{(\F,\T),(\NF, \NF), (\T,\T) \}$.
We can create a program $\mathsf{P}_d$ as follows:
\[
\begin{array}{lll}
  \mbox{\texttt{neg(R)}} & \leftarrow & \pnot\mathtt{R}\\
  \mbox{\texttt{p$_d$(R)}} & \leftarrow & \mbox{\texttt{neg(R)}}\\
  \mbox{\texttt{p$_d$(R)}} & \leftarrow & \mathtt{R}
\end{array}
\]
$\mathsf{P}_d$ is stratified and its unique equilibrium model $\mathcal{M}$ assigns to $\mathtt{p}_d$ the relation $d$.
\end{example}

Consider now the case of \emph{non-total} $\preceq$-monotonic relations. Since our semantics is based on total relations (ie., every equilibrium model assigns to predicate constants total relations), a non-total relation is not directly definable in $\HOL$. However, we can ``capture'' non-total relations in the following sense:  for every non-total $\preceq$-monotonic relation $d$ we can create a stratified $\HOL$ program $\mathsf{P}_d$ which defines a predicate constant $\cst{d}$; the unique equilibrium model ${\cal M}$ of $\mathsf{P}_d$ satisfies ${\cal M}(\cst{d})(\NF) = d$. In order to give the full formal statement of this result, we need the following definition of a collapse function.
\begin{definition}
Let $\pi$ be a predicate type and let $d \in \mo{\pi}$.
We define $\tv{d}$ recursively as:
\begin{itemize}
  \item If $\pi = o$, $\tv{d} = \F$ if $d = \F$, $\tv{d} = \T$ otherwise.
  \item If $\pi = \pi_1\rightarrow \pi_2$, then $\tv{d}(d_1) = \tv{d(d_1)}$ for any $d_1 \in \mo{\pi_1}$.
\end{itemize}
\end{definition}
We can now state our second definability theorem:
\begin{theorem}\label{definability_of_non_exact}
For each type $\pi$ and $d\in\mo{\pi}$, there exists a stratified program $\mathsf{P}_d$ with a constant $\cst{d} : o \rightarrow \pi$ such that
$\mathsf{P}_d$ has a unique equilibrium model $\mathcal{M}$ with
\[
\mathcal{M}(\cst{d})(u)=\begin{cases}
  d, & \text{ if } u = \NF\\
  \tv{d}, & \text{ if } u = \T\\
  \lambda \overline{x}.\F, & \text{ if } u = \F
\end{cases}\]
\end{theorem}
\begin{example}
Let $d \in \mo{o \to o}$ be a relation such that $d=\{(\F,\T),(\NF, \NF), (\T,\NF) \}$.
We can create a program $\mathsf{P}_d$ as follows:
\[
\begin{array}{lll}
  \mbox{\texttt{neg(R)}} & \leftarrow & \pnot\mathtt{R}\\
  \mbox{\texttt{neg\_neg(R)}} & \leftarrow & \mbox{\texttt{neg(neg(R))}}\\
  \mathtt{p}^*_d~\mathtt{U}~\mathtt{R} & \leftarrow & \mbox{\texttt{neg\_neg(U),neg(R)}}\\
  \mathtt{p}^*_d~\mathtt{U}~\mathtt{R} & \leftarrow & \mbox{\texttt{U,neg\_neg(R)}}\\
  \mathtt{p}^*_d~\mathtt{U}~\mathtt{R} & \leftarrow & \mbox{\texttt{U,R}}
\end{array}
\]

It is easy to see that $\mathsf{P}_d$ is stratified and its unique equilibrium model $\mathcal{M}$ satisfies the
requirements of Theorem~\ref{definability_of_non_exact}.
\end{example}

\section{The Strong Equivalence Theorem for Higher-Order Logic Programs}\label{sec:strong-eq}

In this section we present the main theorem of the paper, which establishes that two programs
$\mathsf{P}_1$, $\mathsf{P}_2$ are strongly equivalent iff they have the same models. Our result
is a natural generalization, in programs of arbitrary orders, of the classical strong equivalence
theorem for propositional programs~\citep{LifschitzPV01}. We start by defining the notion of
\emph{strong equivalence} for $\HOL$ programs, which is a direct extension of the corresponding
notion for propositional programs.
\begin{definition}
Let $\mathsf{P}_1,\mathsf{P}_2$ be $\HOL$ programs. Then, $\mathsf{P}_1$ and $\mathsf{P}_2$ will be termed \emph{strongly equivalent}
iff for all $\HOL$ programs $\mathsf{P}$, the programs $\mathsf{P}_1 \cup \mathsf{P}$ and $\mathsf{P}_2 \cup \mathsf{P}$ have the same
equilibrium models.
\end{definition}

Before presenting the theorem, we give two examples of programs that are not strongly equivalent. The first example involves two
propositional programs and therefore we can use directly the classical characterization by~\cite{LifschitzPV01}. The second example
is a related one, involving however a higher-order predicate. Both examples rely on the fact that double negation elimination
is invalid in HT.
\begin{example}
The following example illustrates that the ``program'' $\mathtt{p} \leftarrow \pnot\pnot\mathtt{r}$ is not
strongly equivalent to the program $\mathtt{p} \leftarrow \mathtt{r}$; since our language does not allow nested
negations of the form $\pnot\pnot\mathtt{r}$, we use an intermediate variable \texttt{q} to achieve the same effect.
More specifically, the program $\mathsf{P}_1$:
\[
\begin{array}{l}
\mathtt{p}\leftarrow\pnot\mathtt{q}\\
\mathtt{q}\leftarrow\pnot\mathtt{r}
\end{array}
\]
and the program $\mathsf{P}_2$:
\[
\begin{array}{l}
\mbox{\tt p $\leftarrow$ r}\\
\mbox{\tt q $\leftarrow$ $\pnot$r}
\end{array}
\]
are not strongly equivalent because they don't have the same HT models: the interpretation ${\cal I}(\mathtt{p})=\NF$,
${\cal I}(\mathtt{q})=\F$, and ${\cal I}(\mathtt{r})=\NF$, is a model of $\mathsf{P}_2$ but not of $\mathsf{P}_1$. Actually,
one can easily verify that by following the steps of the proof of Theorem 1 of~\cite{LifschitzPV01}, we can construct the
program $\mathsf{P}=\{\mathtt{p} \leftarrow \mathtt{r}, \mathtt{r} \leftarrow \mathtt{p}\}$ such that $\mathsf{P}_1\cup \mathsf{P}$
and $\mathsf{P}_2\cup \mathsf{P}$ have different equilibrium models. In particular, $\mathsf{P}_1\cup \mathsf{P}$ has the
equilibrium model ${\cal K}(\mathtt{p})=\T$, ${\cal K}(\mathtt{q})=\F$, and ${\cal K}(\mathtt{r})=\T$, which is not an equilibrium
model of $\mathsf{P}_2\cup \mathsf{P}$.
\end{example}
%
%
%
%
\begin{example}
The following example illustrates that a program involving a predicate that is defined through double negation, is not strongly equivalent to
one that does not use double negation. More specifically, the program $\mathsf{P}_1$:
\[
\begin{array}{l}
\mbox{\texttt{p(R)}} \leftarrow \mbox{\texttt{neg(neg(R))}}\\
\mbox{\texttt{neg(R)}} \leftarrow\pnot\mathtt{R}
\end{array}
\]
and the program $\mathsf{P}_2$:
\[
\begin{array}{l}
\mbox{\texttt{p(R)}}\leftarrow\mathtt{R}\\
\mbox{\texttt{neg(R)}} \leftarrow\pnot\mathtt{R}
\end{array}
\]
are not strongly equivalent because they don't have the same higher-order $\HOL$ models. For example,
the interpretation ${\cal I}$ such that ${\cal I}(\mathtt{p})(\F)=\F$, ${\cal I}(\mathtt{p})(\NF)=\NF$,
${\cal I}(\mathtt{p})(\T)=\T$, and ${\cal I}(\mathtt{neg})(\F)=\T$, ${\cal I}(\mathtt{neg})(\NF)=\F$,
${\cal I}(\mathtt{neg})(\T)=\F$, is a model of $\mathsf{P}_2$ but not of $\mathsf{P}_1$. Actually, one
can verify that by taking $\mathsf{P}$ to be the empty set, the programs $\mathsf{P}_1\cup \mathsf{P} = \mathsf{P}_1$
and $\mathsf{P}_2\cup \mathsf{P} = \mathsf{P}_2$ have different equilibrium models. In particular, ${\cal I}$ is an
equilibrium model of $\mathsf{P}_2$ but not of $\mathsf{P}_1$.
\end{example}

The main result of the present paper (Theorem~\ref{se-theorem}), namely that two $\HOL$ programs
are strongly equivalent iff they have the same $\HOL$ models, is a non-trivial generalization of
the corresponding result of~\cite{LifschitzPV01} for propositional programs.
In the rest of this section, we describe, at an intuitive level, the main technical difficulties that arise
in the proof of Theorem~\ref{se-theorem} and how they are handled in our context.
\ifincludeappendix The full proof of the theorem is given in Appendix~\ref{se-appendix}.\else The full proof of the theorem is given in the supplementary material.\fi

What makes the proof of Theorem~\ref{se-theorem} more challenging than the classical one is that in the heart
of the proof lies the need to use the definability results of the previous section. To understand
the need for definability, we have to first give an intuitive description of the central argument in the
proof of the strong equivalence theorem for propositional programs~\citep{LifschitzPV01}: assume that two (zero-order) programs
$\mathsf{P}_1$ and $\mathsf{P}_2$ do not have the same models. For example, assume that there exists an interpretation ${\cal I}$
that is a model of $\mathsf{P}_1$ but not of $\mathsf{P}_2$. How can we find a program $\mathsf{P}$ such that $\mathsf{P}_1 \cup \mathsf{P}$ and $\mathsf{P}_2 \cup \mathsf{P}$ do not have the same equilibrium models? A crucial step in the proof of~\cite{LifschitzPV01} is to include
as facts in $\mathsf{P}$ all atoms $\mathsf{A}$ such that ${\cal I}(\mathsf{A}) = \T$. Turning to the proof for $\HOL$ programs,
assume that for some predicate constant $\mathsf{p}$ in $\mathsf{P}$ and some elements $d_1,\ldots,d_n$ it holds that
${\cal I}(\mathsf{p})\,d_1\cdots d_n = \T$. How could we create a ``higher-order fact'' in $\HOL$ which reflects
that ${\cal I}(\mathsf{p})\,d_1\cdots d_n = \T$?

Consider first the simpler case where ${\cal I}(\mathsf{p})$ \emph{is a total relation}. In this case, we define ${\cal I}(\mathsf{p})$
as a predicate constant $\ce{{\cal I}(\mathsf{p})}$ (see Theorem~\ref{definability_of_exact}) and then add to $\mathsf{P}$ the rule:
$$\mathsf{p}\,\mathsf{R}_1 \cdots \mathsf{R}_n \leftarrow \mbox{$\ce{{\cal I}(\mathsf{p})}\,\mathsf{R}_1 \cdots \mathsf{R}_n$.}$$
For the purposes of the proof of Theorem~\ref{se-theorem}, the above rule behaves as a ``higher-order fact'': intuitively,
given $d_1,\ldots,d_n$ such that ${\cal I}(\mathsf{p})\,d_1\cdots d_n = \T$,
the atom in the body of the above rule forces any interpretation ${\cal J}$ that satisfies the rule to also satisfy
${\cal J}(\mathsf{p})\,d_1\cdots d_n = \T$.

Consider now the more involved case where ${\cal I}(\mathsf{p})$ \emph{is not a total relation}.  For this case the proof uses
Theorem~\ref{definability_of_non_exact}. More specifically, since ${\cal I}(\mathsf{p})$ is not total, there exist total elements $e_1,\ldots,e_n$ such that ${\cal I}(\mathsf{p})\,e_1\cdots e_n = \NF$. We denote by $\mathsf{E}_{U}$ the expression $(\mathsf{p}\,\ce{e_1}\cdots \ce{e_n})$. Then, we create the rule:
$$\mathsf{p}\,\mathsf{R}_1 \cdots \mathsf{R}_n \leftarrow \mbox{$\cst{{\cal I}(\mathsf{p})}\,\mathsf{E}_{U}~\mathsf{R}_1 \cdots \mathsf{R}_n$.}$$
In the context of the proof of Theorem~\ref{se-theorem}, this rule ``behaves'' as a fact. The formal statement of the theorem is given below and \ifincludeappendix its proof can be found in the appendix.\else its proof can be found in the supplementary material.\fi

\begin{theorem}\label{se-theorem}
For any programs $\mathsf{P}_1$ and $\mathsf{P}_2$ the following two statements are equivalent:
\begin{enumerate}
\item For every program $\mathsf{P}$, $\mathsf{P}_1\cup \mathsf{P}$ and $\mathsf{P}_2\cup \mathsf{P}$
      have the same equilibrium models.
\item $\mathsf{P}_1$ and $\mathsf{P}_2$ have the same models.
\end{enumerate}
\end{theorem}
%
%

\section{What is the ``Correct'' Answer-Set Semantics for Higher-Order Programs?}\label{sec:properties}
The study of answer-set semantics for higher-order logic programs was initiated by~\cite{iclp24} using the framework of Approximation Fixpoint Theory (AFT)~\citep{DMT00ApproximationsStableOperatorsWell-FoundedFixpointsApplications,DMT04Ultimateapproximationapplicationnonmonotonicknowledgerepresentation}. A key advantage of AFT is its ability to unify various semantics - such as the well-founded, Kripke-Kleene, and stable model semantics - under a single constructive fixpoint framework. However, AFT also has a limitation: it characterizes these models purely iteratively rather than through a specific logical system. In contrast, the primary strength of the equilibrium semantics is that it offers a \emph{logical characterization}, providing a framework that enables direct logical reasoning about programs.

The divergence between AFT and Equilibrium Logic in the propositional case was
identified several years ago by~\cite{Denecker2012}. While the two techniques yield
identical answer sets for normal propositional logic programs, they differ when
rule bodies contain arbitrary propositional formulas. The precise relationship
between the two is characterized by the following theorem
by~\cite{Denecker2012}.
\begin{theorem}\label{aft-vs-equilibrium}
If no rule body of a propositional logic program $\mathsf{P}$ contains nested negation,
then $M$ is a stable model of $\mathsf{P}$ in the
sense of AFT iff $M$ is an equilibrium model of $\mathsf{P}$.
\end{theorem}

Intuitively, the AFT-based semantics reduces double negation to identity, whereas the equilibrium semantics does not.
\begin{example}\label{first-difference}
The program \mbox{\texttt{p$\leftarrow\pnot\pnot$p}} has only the model $\{(\texttt{p},\F)\}$ in AFT while it has two equilibrium
models namely $\{(\texttt{p},\F)\}$ and $\{(\texttt{p},\T)\}$.
\end{example}

Unfortunately, the simple syntactic characterization of the differences between AFT and the equilibrium semantics, as implied by Theorem~\ref{aft-vs-equilibrium}, does not straightforwardly generalize to $\HOL$ programs. In a higher-order context, the negation operator can effectively ``disguise'' itself, as illustrated by the following example.
\begin{example}
The following program does not, at first sight, contain nested negation (the negation operator is hidden behind a different name).
\[
\begin{array}{l}
\mbox{\texttt{p}} \leftarrow \mbox{\texttt{neg(neg(p))}}\\
\mbox{\texttt{neg(R)}} \leftarrow\pnot\mathtt{R}
\end{array}
\]
Again, this program has one stable model in AFT and two equilibrium models.
\end{example}

Because negation can be disguised in $\HOL$ in numerous ways (for instance, through a predicate that uses irrelevant code),
syntactically characterizing the class of $\HOL$ programs where nested negation is used, appears to be a non-trivial,
if not impossible, task. Crucially, double negation is not the sole source of conflict. One can define alternative predicates in $\HOL$ which, when used in a recursive context, also produce differing outcomes under AFT and equilibrium semantics.
The following is such an example.
\begin{example}\label{second-vicious-example}
Consider the program:
\[
\begin{array}{l}
\mbox{\texttt{q(Q)}} \leftarrow \mbox{\texttt{Q}}\\
\mbox{\texttt{q(Q)}} \leftarrow \mbox{\texttt{$\pnot$Q}}\\
\mbox{\texttt{p}} \leftarrow\mbox{\texttt{q(p)}}
\end{array}
\]
It can be verified that in AFT the program has a unique stable model in which \texttt{p} is $\T$; on the other hand, it does not have any equilibrium models.
\end{example}

These examples naturally raise the question: what is the ``correct'' answer-set semantics for higher-order logic programs? If one seeks generality and a constructive (fixpoint) approach for obtaining the semantics of a program, AFT is more suitable. However, we argue below that the equilibrium semantics offers a distinct advantage when the goal is to reason about programs and justify program transformations.
\begin{example}
Consider again the program of Example~\ref{second-vicious-example}. Assume we want to unfold the definition of \texttt{q} in the
body of the definition of \texttt{p}. This would give the new program:
\[
\begin{array}{l}
\mbox{\texttt{q(Q)}} \leftarrow \mbox{\texttt{Q}}\\
\mbox{\texttt{q(Q)}} \leftarrow \mbox{\texttt{$\pnot$Q}}\\
\mbox{\texttt{p}} \leftarrow\mbox{\texttt{p}}\\
\mbox{\texttt{p}} \leftarrow\mbox{\texttt{$\pnot$p}}
\end{array}
\]
Although the program seems, from a programmer's point of view, to be equivalent to that of Example~\ref{second-vicious-example}, AFT
now does not produce any stable model. On the other hand, the equilibrium semantics retains its behaviour: the new program does not
have any equilibrium models (as it was also the case for the initial one).
\end{example}

The above examples illustrate that mixing recursion and predicate application in unrestricted ways, results in discrepancies between the AFT and the equilibrium semantics. However, if recursion and applications are used in a controlled way, the two techniques converge.
\ifincludeappendix The proof of the following theorem is given in the appendix.\else The proof of the following theorem is given in the supplementary material.\fi
\begin{theorem}\label{aft-coincide-equilibrium-on-stratified}
Let $\mathsf{P}$ be a stratified program.
Let ${\cal M}$ be its unique equilibrium model and ${\cal N}$ be its unique 2-valued stable model under the AFT semantics.
Then, $\pi({\cal M}) = {\cal N}$, where $\pi$ is a collapse function from ${\cal H}$
to the set of 2-valued interpretations.
\end{theorem}
We believe there exist much broader classes of programs than the stratified ones where the AFT-based and the equilibrium semantics
coincide. One such potential fragment is the language ``\emph{Stratified+Choices Higher-Order
Datalog$^\neg$}'' defined in~\citep{iclp25}, in which programs, intuitively, consist of a stratified part residing on top
of a choice program. We have the following conjecture:
\begin{conjecture}
Let $\mathsf{P}$ be a program of Stratified+Choices Higher-Order Datalog$^\neg$. Then, ${\cal M}$ is an equilibrium model
of $\mathsf{P}$ iff $\pi({\cal M})$ is a 2-valued stable model of $\mathsf{P}$ under AFT,  where $\pi$ is a collapse function
from ${\cal H}$ to the set of 2-valued interpretations.
\end{conjecture}

Beyond the aforementioned conjecture - which is significant for reconciling two prominent semantic frameworks - this
work suggests several promising directions for future research. We briefly outline a few of these. A natural progression
would be to extend equilibrium semantics to accommodate nested implications in rules. From a practical standpoint,
it is also crucial to explore an extension of $\HOL$ that supports an individual data type $\iota$ (in addition to
the Boolean type $o$), potentially building upon ideas found in~\cite{Pearce2008,Lifschitz2007}.
Another interesting direction
is the extension of temporal ASP~\citep{temporal_cabalar} and stream reasoning~\citep{stream_reasoning} to the higher-order context.

Finally, future work
could investigate the logical foundations of the well-founded semantics for $\HOL$ programs, utilizing a higher-order
extension of the two-dimensional HT logic introduced by~\cite{Cabalar01}.

\bibliography{ht}
\ifincludeappendix
\clearpage
\appendix
\input{proofs}

\fi

\end{document}

%% file: titleinfo.tex
\lefttitle{A. Charalambidis, G. Chatziagapis, B. Kostopoulos and P. Rondogiannis}

\jnlPage{1}{8}
\jnlDoiYr{2021}
\doival{10.1017/xxxxx}

\title{Equilibrium Semantics and Strong Equivalence for Higher-Order Logic Programs}


\begin{authgrp}
\author{%
      \gn{Angelos} \sn{Charalambidis}\\
      \affiliation{Harokopio University of Athens, Greece}
      \email{acharal@hua.gr}
}
\author{%
      \gn{Giannos} \sn{Chatziagapis}\\
\affiliation{National and Kapodistrian University of Athens, Greece}
      \email{gchatziagap@di.uoa.gr}
}
\author{%
\gn{Babis} \sn{Kostopoulos}\\
\affiliation{Harokopio University of Athens, Greece}
\email{kostbabis@hua.gr}
}
\author{
      \gn{Panos} \sn{Rondogiannis}\\
      \affiliation{National and Kapodistrian University of Athens, Greece}
      \email{prondo@di.uoa.gr}
}
\end{authgrp}


\jnlPage{1}{14}
\jnlDoiYr{2024}
\doival{10.1017/xxxxx}

\history{\sub{xx xx xxxx;} \rev{xx xx xxxx;} \acc{xx xx xxxx}}

%% file: proofs.tex
\section{Proofs of Section~\ref{sec:semantics_of_HOL}}\label{appendix-of-section-3}
The following is an example $\HOL$ program producing, as distinct equilibrium models, all the
``two-valued'' $\preceq$-monotonic relations of type $o \to o$.
\begin{example}
Example~\ref{two-valued-monotone} demonstrates the production of all ``two-valued'' $\preceq$-monotonic
relations of type $o \to o$. Can we produce \emph{all the total} $\preceq$-monotonic ones? This is
clearly more difficult, as the following program attests.
\[
\begin{array}{lcl}
  \mbox{\texttt{neg(R)}} & \leftarrow & \pnot\mathtt{R}\\
  \mathtt{tt} & \leftarrow & \pnot  \overline{\texttt{tt}}\\
  \overline{\mathtt{tt}} & \leftarrow & \pnot  \texttt{tt}\\
  \mathtt{ff} & \leftarrow & \pnot  \overline{\texttt{ff}}\\
  \overline{\mathtt{ff}} & \leftarrow & \pnot  \texttt{ff}\\
  \mathtt{tst} & \leftarrow & \mathtt{tt,} \pnot  \overline{\texttt{tst}}\\
  \overline{\mathtt{tst}} & \leftarrow & \pnot  \texttt{tst}\\
  \mbox{\texttt{p(R)}} & \leftarrow & \texttt{tt}, \texttt{R}\\
  \mbox{\texttt{p(R)}} & \leftarrow & \texttt{ff}, \pnot \texttt{R}\\
  \mbox{\texttt{p(R)}} & \leftarrow & \texttt{tst}, \texttt{neg(neg(R))}
\end{array}
\]
Intuitively, the equilibrium models (see Table 1) are formed by the valid combinations of the choice atoms \texttt{ff}, \texttt{tt}, and \texttt{tst}. Note that \texttt{tst} cannot be true if \texttt{tt} is false.
\begin{table}[ht]
    \centering
    \caption{The 6 equilibrium models corresponding to the total $\preceq$-monotonic relations.}
    \label{tab:equilibrium_models}
    \begin{tabular}{lll|lll}
        \toprule
        \multicolumn{3}{c}{\textbf{Choice Atoms}} & \multicolumn{3}{c}{\textbf{Relation Values}} \\
        \cmidrule(lr){1-3} \cmidrule(lr){4-6}
        \texttt{ff} & \texttt{tt} & \texttt{tst} & $p(\F)$ & $p(\NF)$ & $p(\T)$ \\
        \midrule
        \F & \F & \F & \F & \F & \F \\
        \T & \F & \F & \T & \F & \F \\
        \F & \T & \F & \F & \NF & \T \\
        \T & \T & \F & \T & \NF & \T \\
        \F & \T & \T & \F & \T & \T \\
        \T & \T & \T & \T & \T & \T \\
        \bottomrule
    \end{tabular}
\end{table}

\end{example}
\begin{relemma}{splitting_submodel}
Let $\mathsf{P}$ be a program and $U$ be a splitting set of $\mathsf{P}$.
If $\mathcal{M}$ is an equilibrium model of $\mathsf{P}$,
then $\mathcal{M}$ restricted to $U$ is an equilibrium model of $b_U(\mathsf{P})$.
\end{relemma}
\begin{proof}
Let $\mathcal{M'}$ be the restriction of $\mathcal{M}$ to $U$.
Since $\mathcal{M}$ is a model of $\mathsf{P}$, and $b_U(\mathsf{P}) \subseteq \mathsf{P}$,
$\mathcal{M'}$ is a model of $b_U(\mathsf{P})$.
It remains to show that $\mathcal{M'}$ is $\preceq$-minimal.
Suppose, for the sake of contradiction, that there exists a model $\mathcal{N'}$ of $b_U(\mathsf{P})$ such that $\mathcal{N'} \prec \mathcal{M'}$.
We define the interpretation $\mathcal{N}$ of $\mathsf{P}$ as follows:
\[
\mathcal{N}(\mathsf{p}) = \begin{cases}
  \mathcal{N'}(\mathsf{p}), & \text{ if } \mathsf{p}\in U\\
  \mathcal{M}(\mathsf{p}), & \text{ otherwise}
\end{cases}
\]
Notice that $\mathcal{N} \prec \mathcal{M}$. Let $C = \mathsf{H} \leftarrow \mathsf{B}$ where $\mathsf{H}=\mathsf{p}~\mathsf{R}_1~\cdots~\mathsf{R}_n$ be a clause of $\mathsf{P}$ and $s$ be a state.
If $\mathsf{p}\in U$, then for all predicate constants $\mathsf{q}$ appearing in $C$, $\mathsf{q}$ would also be in $U$, so that $\mathcal{N}(\mathsf{q}) = \mathcal{N'}(\mathsf{q})$.
Since $\mathcal{N'}$ is a model of $b_U(\mathsf{P})$, $\mathcal{N}$ would satisfy $C$.
If $\mathsf{p} \not\in U$, we have $\mathcal{N}(\mathsf{p}) = \mathcal{M}(\mathsf{p})$, so that $\mwrs{\mathsf{H}}{\mathcal{N}}{s} = \mwrs{\mathsf{H}}{\mathcal{M}}{s}$.
Also, since $\mathcal{N} \prec \mathcal{M}$, we have $\mwrs{\mathsf{B}}{\mathcal{N}}{s} \preceq \mwrs{\mathsf{B}}{\mathcal{M}}{s}$.
If $\mwrs{\mathsf{B}}{\mathcal{M}}{s} = \T$, since $\mathcal{M}$ is a model of $\mathsf{P}$,  $\mwrs{\mathsf{H}}{\mathcal{M}}{s} = \T$.
So, we have $\mwrs{\mathsf{B}}{\mathcal{N}}{s} \in \{\T,\NF\}$ and $\mwrs{\mathsf{H}}{\mathcal{N}}{s} = \T$.
The case is similar for $\mwrs{\mathsf{B}}{\mathcal{M}}{s} = \NF$ and trivial for $\mwrs{\mathsf{B}}{\mathcal{M}}{s} = \F$.
In any case, $\mathcal{N}$ satisfies $C$, so that $\mathcal{N}$ is a model of $\mathsf{P}$.
But $\mathcal{N} \prec \mathcal{M}$ and $\mathcal{M}$ is an equilibrium model of $\mathsf{P}$, which is a contradiction.
\end{proof}

We now give a detailed proof of the following statement, stated in Section~\ref{sec:semantics_of_HOL}:
\begin{retheorem}{is_unique}
Let $\mathsf{P}$ be a stratified $\HOL$ program. Then, $\mathsf{P}$ has a unique equilibrium model.
\end{retheorem}

It is easy to see that for any program $\mathsf{P}$ with stratification function $S$, and every natural number
$n$, the set $\{\mathsf{p} \mid S(\mathsf{p})< n\}$ is a splitting set of $\mathsf{P}$. Using these splitting sets
we can divide $\mathsf{P}$ into a finite number of strata $\mathsf{P}_1,\ldots,\mathsf{P}_k$.

In the rest of this section we consider a stratified higher-order logic program $\mathsf{P}$ which we
will denote by $\mathsf{P}_1 \cup \cdots \cup \mathsf{P}_k$, assuming that $\mathsf{P}_1$ is the lowest stratum
and $\mathsf{P}_k$ the highest one. The proof is based on constructing, in a bottom-up way, the unique equilibrium
model of this program. We start by defining an alternative evaluation function for stratified $\HOL$ programs.
\begin{definition}\label{tuple-semantics}
Let $\mathsf{P} = \mathsf{P}_1 \cup \cdots \cup \mathsf{P}_k$ be a stratified $\HOL$ program and
let $i \in \{1,\ldots,k\}$. Let ${\cal I}, {\cal J}$ be partial interpretations of $\mathsf{P}$, where ${\cal I}$ is
restricted on the predicate constants defined in $\mathsf{P}_i$ and ${\cal J}$ on those defined in
$\mathsf{P}_1 \cup \cdots \cup \mathsf{P}_{i-1}$. We define an evaluation function $\lsem \cdot \rsem^*$ for
the expressions that appear in the bodies of rules of $\mathsf{P}_i$, as follows:
\begin{enumerate}
  \item $\mwrst{(\mathsf{L}_1 \wedge \cdots \wedge \mathsf{L}_m)}{{\cal I},{\cal J}}{s} =
    min_\leq\{\mwrst{\mathsf{L}_1}{{\cal I},{\cal J}}{s},\ldots,\mwrst{\mathsf{L}_m}{{\cal I},{\cal J}}{s}\}$



  \item $\mwrst{\mathsf{q}\, \mathsf{E}_1\cdots \mathsf{E}_r}{{\cal I},{\cal J}}{s} =
         {\cal I}(\mathsf{q})\,\mwrs{\mathsf{E}_1}{{\cal J}}{s}\cdots \mwrs{\mathsf{E}_r}{{\cal J}}{s}$

  \item $\mwrst{\,\sim \!\mathsf{q}\, \mathsf{E}_1\cdots \mathsf{E}_r}{{\cal I},{\cal J}}{s} =\,
         \sim \! {\cal J}(\mathsf{q})\,\mwrs{\mathsf{E}_1}{{\cal J}}{s}\cdots \mwrs{\mathsf{E}_r}{{\cal J}}{s}$

  \item $\mwrst{\mathsf{Q}\, \mathsf{E}_1\cdots \mathsf{E}_r}{{\cal I},{\cal J}}{s} =
         s(\mathsf{Q})\,\mwrs{\mathsf{E}_1}{{\cal J}}{s}\cdots \mwrs{\mathsf{E}_r}{{\cal J}}{s}$

  \item $\mwrst{\,\sim \!\mathsf{Q}\, \mathsf{E}_1\cdots \mathsf{E}_r}{{\cal I},{\cal J}}{s} =
         \,\sim \!s(\mathsf{Q})\,\mwrs{\mathsf{E}_1}{{\cal J}}{s}\cdots \mwrs{\mathsf{E}_r}{{\cal J}}{s}$
 \end{enumerate}
\end{definition}

\vspace{0.2cm}
\noindent
\begin{definition}
Let $\mathsf{P} = \mathsf{P}_1 \cup \cdots \cup \mathsf{P}_k$ be a stratified $\HOL$ program,
let $i \in \{1,\ldots,k\}$ and let ${\cal J}$ be a partial interpretation of $\mathsf{P}$ restricted on
the predicate constants defined in $\mathsf{P}_1 \cup \cdots \cup \mathsf{P}_{i-1}$. Then, for every predicate constant
$\mathsf{p} : \pi_1 \to \cdots \to \pi_n \to \bool$ defined in $\mathsf{P}_i$, for all $d_1 \in \mo{\pi_1},\ldots,
d_n \in \mo{\pi_n}$ and for every interpretation ${\cal I}$ restricted on the predicate constants of $\mathsf{P}_i$,
the immediate consequence operator $T_{\mathsf{P}_i,\cal J}$ of $\mathsf{P}_i$ is defined as follows: \\
  $${\cal T}_{\mathsf{P}_i,\cal J}({\cal I})(\mathsf{p})\ \overline{d} =
        max_{\leq}\{
          \mwrst{\mathsf{B}}{{\cal I},{\cal J}}{s[\overline{\mathsf{R}}/\overline{d}]} \mid
                  \mbox{$s\in {\cal S}$ and
                        $(\mathsf{p}\ \overline{\mathsf{R}} \lrule \mathsf{B})$ in $\mathsf{P}_i$}\}.$$
\end{definition}
\begin{lemma}\label{T_P_monotonicity}
Let $\mathsf{P} = \mathsf{P}_1 \cup \cdots \cup \mathsf{P}_k$ be a stratified $\HOL$ program,
let $i \in \{1,\ldots,k\}$ and let ${\cal J}$ be a partial interpretation of $\mathsf{P}$ restricted on
the predicate constants defined in $\mathsf{P}_1 \cup \cdots \cup \mathsf{P}_{i-1}$.
Let ${\cal I}_1,{\cal I}_2$ be partial interpretations of $\mathsf{P}$ restricted on the predicate constants
defined in $\mathsf{P}_i$. Assume that ${\cal I}_1 \leq {\cal I}_2$. Then,
${\cal T}_{\mathsf{P}_i,{\cal J}}({\cal I}_1) \leq {\cal T}_{\mathsf{P}_i,{\cal J}}({\cal I}_2)$.
\end{lemma}
\begin{proof}
Let $\mathsf{p} : \pi_1 \to \cdots \to \pi_n \to \bool$ be a predicate constant defined in $\mathsf{P}_i$, and $d_1 \in \mo{\pi_1},\ldots,
d_n \in \mo{\pi_n}$. We show that ${\cal T}_{\mathsf{P}_i,{\cal J}}({\cal I}_1)(\mathsf{p})\ d_1 \cdots d_n \leq
{\cal T}_{\mathsf{P}_i,{\cal J}}({\cal I}_2)(\mathsf{p})\ d_1 \cdots d_n$.
It suffices to show that for every $\mathsf{L}_i$ in the body of any rule for $\mathsf{p}$ and for
every state $s$, it holds that
$\lsem \mathsf{L}_i \rsem^{*}_{s'} ({\cal I}_1,{\cal J}) \leq  \lsem \mathsf{L}_i \rsem^{*}_{s'} ({\cal I}_2,{\cal J})$,
where $s' = s[\mathsf{R}_1/d_1,\ldots,\mathsf{R}_n/d_n]$. We distinguish cases:
\begin{itemize}
%

\item $\mathsf{L}_i = \mathsf{q}\, \mathsf{E}_1 \cdots \mathsf{E}_r$. But then,
      $\lsem \mathsf{L}_i \rsem^{*}_{s'} ({\cal I}_1,{\cal J}) =
      {\cal I}_1(\mathsf{q})\, \lsem \mathsf{E}_1\rsem_{s'}({\cal J}) \cdots
      \lsem \mathsf{E}_r\rsem_{s'}({\cal J}) \leq
      {\cal I}_2(\mathsf{q})\, \lsem \mathsf{E}_1\rsem_{s'}({\cal J}) \cdots
      \lsem \mathsf{E}_r\rsem_{s'}({\cal J}) = \lsem \mathsf{L}_i \rsem^{*}_{s'} ({\cal I}_2,{\cal J})$, using the
      definition of $\lsem \cdot \rsem^*$ and the fact that ${\cal I}_1 \leq {\cal I}_2$.

\item $\mathsf{L}_i =\,\sim\!\!\mathsf{q}\, \mathsf{E}_1 \cdots \mathsf{E}_r$. But then,
      $\lsem \mathsf{L}_i \rsem^{*}_{s'} ({\cal I}_1,{\cal J}) =\,\sim\!\!{\cal J}(\mathsf{q})\, \lsem \mathsf{E}_1\rsem_{s'}({\cal J}) \cdots \lsem \mathsf{E}_r\rsem_{s'}({\cal J}) = \lsem \mathsf{L}_i \rsem^{*}_{s'} ({\cal I}_2,{\cal J})$, using the
      definition of $\lsem \cdot \rsem^*$.

\item $\mathsf{L}_i = \mathsf{Q} \,\mathsf{E}_1 \cdots \mathsf{E}_r$. But then,
      $\lsem \mathsf{L}_i \rsem^{*}_{s'} ({\cal I}_1,{\cal J}) = s(\mathsf{Q})\,\mwrs{\mathsf{E}_1}{{\cal J}}{s}\cdots \mwrs{\mathsf{E}_r}{{\cal J}}{s} = \lsem \mathsf{L}_i \rsem^{*}_{s'} ({\cal I}_2,{\cal J})$, using the definition of $\lsem \cdot \rsem^*$.

\item $\mathsf{L}_i = \,\sim \!\!\mathsf{Q} \,\mathsf{E}_1 \cdots \mathsf{E}_r$. But then,
      $\lsem \mathsf{L}_i \rsem^{*}_{s'} ({\cal I}_1,{\cal J}) = \,\sim \!\!s(\mathsf{Q})\,\mwrs{\mathsf{E}_1}{{\cal J}}{s}\cdots \mwrs{\mathsf{E}_r}{{\cal J}}{s} = \lsem \mathsf{L}_i \rsem^{*}_{s'} ({\cal I}_2,{\cal J})$, using the definition of $\lsem \cdot \rsem^*$.
\end{itemize}
This completes the proof of the lemma.
\end{proof}
We now define a sequence of partial interpretations of $\mathsf{P}$ that, as we are going to show, lead to the construction
of the unique equilibrium model of $\mathsf{P}$.
\begin{definition}
Let $\mathsf{P} = \mathsf{P}_1 \cup \cdots \cup \mathsf{P}_k$ be a stratified $\HOL$ program. We define the
following sequence of partial interpretations of $\mathsf{P}$:
$${\cal M}_1 = \mathit{lfp}({\cal T}_{\mathsf{P}_1,\emptyset}),
{\cal M}_2 = {\cal M}_1 \cup \mathit{lfp}({\cal T}_{\mathsf{P}_2,{\cal M}_1}),\ldots,
{\cal M}_k = {\cal M}_{k-1} \cup \mathit{lfp}({\cal T}_{\mathsf{P}_{k},{\cal M}_{k-1}}).$$
\end{definition}

Notice that the least fixpoints in the above definition exist due to the fact that the set of interpretations
is a complete lattice and, by Lemma~\ref{T_P_monotonicity}, the ${\cal T}_{\mathsf{P}_i,{\cal M}_{i-1}}$ operators are monotonic.

In the rest of this appendix section, we will use freely the above partial interpretations ${\cal M}_i$ of
$\mathsf{P}$ in our statements.
\begin{lemma}\label{modelhood}
Let $\mathsf{P} = \mathsf{P}_1 \cup \cdots \cup \mathsf{P}_k$ be a stratified $\HOL$
program and let $i \in \{1,\ldots,k\}$. Let ${\cal N}$ be an interpretation over the predicate constants
defined in $\mathsf{P}_i$ and assume that ${\cal M}_{i-1} \cup {\cal N}$ is a model of $\mathsf{P}_i$.
Then, ${\cal T}_{\mathsf{P}_{i},{\cal M}_{i-1}}({\cal N}) \leq {\cal N}$.
\end{lemma}
\begin{proof}
Let $\mathsf{p} \, \mathsf{R}_1\cdots \mathsf{R}_n \leftarrow \mathsf{B}$
be a rule in $\mathsf{P}_i$. Since ${\cal M}_{i-1} \cup {\cal N}$ is a model of $\mathsf{P}_i$,
it holds that ${\cal N}(\mathsf{p})\,d_1 \cdots d_n \geq \mwrs{\mathsf{B}}{{\cal M}_{i-1} \cup {\cal N}}{s[\overline{\mathsf{R}}/\overline{d}]}$. It is easy to check that
$\mwrs{\mathsf{B}}{{\cal M}_{i-1} \cup {\cal N}}{s[\overline{\mathsf{R}}/\overline{d}]} = \mwrst{\mathsf{B}}{{\cal N},{\cal M}_{i-1}}{s[\overline{\mathsf{R}}/\overline{d}]}$.
But then, ${\cal N}(\mathsf{p})\,d_1 \cdots d_n \geq \mwrst{\mathsf{B}}{{\cal N},{\cal M}_{i-1}}{s[\overline{\mathsf{R}}/\overline{d}]}$, which implies that ${\cal N}(\mathsf{p})\,d_1 \cdots d_n \geq \bigvee_{\leq_\bool}\{
\mwrst{\mathsf{B}}{{\cal N},{\cal M}_{i-1}}{s[\overline{\mathsf{R}}/\overline{d}]} \mid
\mbox{$s\in {\cal S}$ and $(\mathsf{p}\ \overline{\mathsf{R}} \lrule \mathsf{B})$ in $\mathsf{P}_i$}\}$.
Thus, ${\cal N}(\mathsf{p})\,d_1 \cdots d_n \geq {\cal T}_{\mathsf{P}_{i},{\cal M}_{i-1}}({\cal N})(\mathsf{p})\,d_1 \cdots d_n$,
which implies that ${\cal N} \geq {\cal T}_{\mathsf{P}_{i},{\cal M}_{i-1}}({\cal N})$.
\end{proof}
\begin{lemma}\label{minimality}
Let $\mathsf{P} = \mathsf{P}_1 \cup \cdots \cup \mathsf{P}_k$ be a stratified $\HOL$
program and let $i\in \{1,\ldots,k\}$. Then, ${\cal M}_i$ is the $\leq$-minimum among all
models of $\mathsf{P}_1 \cup \cdots \cup \mathsf{P}_{i}$ that agree with ${\cal M}_{i-1}$.
Moreover, ${\cal M}_i$ is a $\leq$-minimal model of $\mathsf{P}_1 \cup \cdots \cup \mathsf{P}_{i}$.
\end{lemma}
\begin{proof}
Let ${\cal N}$ be another model of $\mathsf{P}_1 \cup \cdots \cup \mathsf{P}_{i}$
that agrees with ${\cal M}_{i-1}$. Since ${\cal M}_i$ and ${\cal N}$ agree on the predicate constants defined in
${\cal M}_{i-1}$, it suffices to show that $\mathit{lfp}({\cal T}_{\mathsf{P}_{i},{\cal M}_{i-1}}) \leq \overline{\cal N}$, where
$\overline{\cal N}$ is the restriction of ${\cal N}$ on the predicate constants defined in $\mathsf{P}_i$.
The proof is by induction on the approximations to the least fixed point. The basis case is trivial.
Assume that ${\cal T}^j_{\mathsf{P}_{i},{\cal M}_{i-1}}(\perp) \leq \overline{\cal N}$. Then,
using monotonicity and Lemma~\ref{modelhood}, we have that
${\cal T}_{\mathsf{P}_{i},{\cal M}_{i-1}}({\cal T}^j_{\mathsf{P}_{i},{\cal M}_{i-1}}(\perp)) \leq
{\cal T}_{\mathsf{P}_{i},{\cal M}_{i-1}}(\overline{\cal N}) \leq \overline{\cal N}$.

To show the second part of the statement of the lemma, we use induction on $i$. For $i=1$ the result follows
from the proof of the first statement of the current lemma. For the induction step,  assuming that ${\cal M}_{i-1}$ is a
minimal model of $\mathsf{P}_1 \cup \cdots \cup \mathsf{P}_{i-1}$, we show that ${\cal M}_{i}$ is a
minimal model of $\mathsf{P}_1 \cup \cdots \cup \mathsf{P}_{i}$. Assume, for the sake of contradiction,
that ${\cal M}_{i}$ is not a minimal model of $\mathsf{P}_1 \cup \cdots \cup \mathsf{P}_{i}$.
Then, there exists an interpretation ${\cal N}$ that is a model of $\mathsf{P}_1 \cup \cdots \cup \mathsf{P}_{i}$
and ${\cal N}<{\cal M}_i$. Notice that since ${\cal M}_{i-1}$ is a minimal model of
$\mathsf{P}_1 \cup \cdots \cup \mathsf{P}_{i-1}$, then ${\cal N}$ coincides with ${\cal M}_{i-1}$ on all
predicates defined in $\mathsf{P}_1 \cup \cdots \cup \mathsf{P}_{i-1}$. Let ${\cal N}_i$ be the restriction
of ${\cal N}$ on the predicate constants defined in $\mathsf{P}_i$. We show by an inner induction that for all $n\geq 0$,
${\cal T}^n_{\mathsf{P}_{i},{\cal M}_{i-1}}({\cal N}_i) \leq {\cal N}_i$. For $n=0$, the statement obviously holds.
Assuming that ${\cal T}^n_{\mathsf{P}_{i},{\cal M}_{i-1}}({\cal N}_i) \leq {\cal N}_i$ and using the monotonicity
of ${\cal T}_{\mathsf{P}_i}$, the induction hypothesis, and Lemma~\ref{modelhood}, we get that:
$${\cal T}_{\mathsf{P}_{i},{\cal M}_{i-1}}({\cal T}^n_{\mathsf{P}_{i},{\cal M}_{i-1}}({\cal N}_i)) \leq
{\cal T}_{\mathsf{P}_{i},{\cal M}_{i-1}}({\cal N}_i) \leq {\cal N}_i$$
Therefore, $\mathit{lfp}({\cal T}_{\mathsf{P}_{i},{\cal M}_{i-1}}) \leq {\cal N}_i$ and thus
${\cal M}_i = {\cal M}_{i-1} \cup  \mathit{lfp}({\cal T}_{\mathsf{P}_{i},{\cal M}_{i-1}}) \leq {\cal N}$, which is
a contradiction.
\end{proof}

Before proceeding to establish that stratified $\HOL$ programs have unique
equilibrium models, we need to characterize a broad class of positive $\HOL$ programs that have unique such models.
It is important to note that, in contrast to positive first-order logic programs, not all positive $\HOL$ programs have
unique equilibrium models. However, the definition of stratification, implicitly defines a class of positive
programs that are well-behaved in this respect. The following two definitions formalize this class of programs.
\begin{definition}
A {\em simple positive atom} is either:
\begin{itemize}
%
\item an atom of the form $\mathsf{p}\, \mathsf{E}_1 \cdots \mathsf{E}_r$ where the $\mathsf{E}_i$'s
      contain only predicate variables, or
\item an atom of the form $\mathsf{R} \,\mathsf{E}_1 \cdots \mathsf{E}_r$ where the $\mathsf{E}_i$'s
      contain only predicate variables.
\end{itemize}
\end{definition}

\begin{definition}
A program $\mathsf{P}$ of $\HOL$ will be called {\em simple positive} iff for every rule
$\mathsf{p} \, \mathsf{R}_1\cdots \mathsf{R}_n \leftarrow \mathsf{L}_1,\ldots,\mathsf{L}_m$ in $\mathsf{P}$,
each $\mathsf{L}_i$ is a simple positive atom.
\end{definition}
\begin{lemma}\label{sp_uniqueness}
Let $\mathsf{P}$ be a simple positive logic program. Then, $\mathsf{P}$ has a unique equilibrium model.
\end{lemma}
\begin{proof}
The proof follows very similar steps as the proofs of the induction steps of the subsequent Lemma~\ref{is_equilibrium} and
Theorem~\ref{is_unique} (but obviously the proof of the present lemma does not require an induction hypothesis, since simple
positive logic programs correspond to the lowest stratum of stratified $\HOL$ programs). We omit the details
of the proof due to the similarity to the subsequent ones.
\end{proof}

\begin{lemma}\label{is_equilibrium}
Let $\mathsf{P} = \mathsf{P}_1 \cup \cdots \cup \mathsf{P}_k$ be a stratified $\HOL$
program and let $i\in \{1,\ldots,k\}$. Then, ${\cal M}_i$ is an equilibrium
model of $\mathsf{P}_1 \cup \cdots \cup \mathsf{P}_{i}$.
\end{lemma}
\begin{proof}
The proof is by induction on $i$. The basis case is for $i=1$ and follows from Lemma~\ref{sp_uniqueness}.
For $i>1$, assuming that ${\cal M}_{i-1}$ is an equilibrium
model of $\mathsf{P}_1 \cup \cdots \cup \mathsf{P}_{i-1}$, we show that ${\cal M}_{i}$ is an
equilibrium model of $\mathsf{P}_1 \cup \cdots \cup \mathsf{P}_{i}$.
Since ${\cal M}_{i}$ is $\leq$-minimal, it is also $\preceq$-minimal. It suffices to show that it is total.
Let $d_1,\ldots,d_n$ be total elements. We use an inner induction on the approximations of the least fixpoint:
we show that for every $j$, ${\cal T}^j_{\mathsf{P}_{i},{\cal M}_{i-1}}(\perp)(\mathsf{p})\, d_1\cdots d_n$ is total.
For $j=0$ the proof is immediate.
It suffices to show that ${\cal T}^{j+1}_{\mathsf{P}_{i},{\cal M}_{i-1}}(\perp)(\mathsf{p})\, d_1\cdots d_n$ is total.
Let $\mathsf{p}\,\mathsf{R}_1 \cdots \mathsf{R}_n \leftarrow \mathsf{B}$ be a rule in $\mathsf{P}$.
Let $s' = s[\mathsf{R}_1/d_1,\ldots,\mathsf{R}_n/d_n]$ and $s''=s'[\mathsf{Q}_1/e_1,\ldots,\mathsf{Q}_l/e_l]$,
where the $\mathsf{Q}_i$ are the existential variables in $\mathsf{B}$ and $e_1,\ldots,e_l$ are total elements
such that $s'(\mathsf{Q}_i) \preceq e_i$. Therefore, $s' \preceq s''$.
Notice that $\lsem \mathsf{B}\rsem^{*}_{s'} ({\cal T}^j_{\mathsf{P}_{i},{\cal M}_{i-1}}(\perp),{\cal M}_{i-1}) \preceq \lsem \mathsf{B}\rsem^{*}_{s''} ({\cal T}^j_{\mathsf{P}_{i},{\cal M}_{i-1}}(\perp),{\cal M}_{i-1})$
and therefore $\lsem \mathsf{B}\rsem^{*}_{s'} ({\cal T}^j_{\mathsf{P}_{i},{\cal M}_{i-1}}(\perp),{\cal M}_{i-1}) \leq \lsem \mathsf{B}\rsem^{*}_{s''} ({\cal T}^j_{\mathsf{P}_{i},{\cal M}_{i-1}}(\perp),{\cal M}_{i-1})$.
Therefore, for every rule, there exists a total state under which it receives its $\leq$-maximum value.
Moreover, this value is always total, as we now show by a case analysis on the literals that exist in $\mathsf{B}$:
\begin{itemize}
%

\item $\mathsf{L} = \mathsf{q}\, \mathsf{E}_1 \cdots \mathsf{E}_r$. Then:
      $$\lsem \mathsf{L}\rsem^{*}_{s''} ({\cal T}^j_{\mathsf{P}_{i},{\cal M}_{i-1}}(\perp),{\cal M}_{i-1}) = {\cal T}^j_{\mathsf{P}_{i},{\cal M}_{i-1}}(\perp)(\mathsf{q}) \, \mwrs{\mathsf{E}_1}{{\cal M}_{i-1}}{s''}\cdots \mwrs{\mathsf{E}_r}{{\cal M}_{i-1}}{s''}$$
      which is total by the inner induction hypothesis and the fact that all the $\mwrs{\mathsf{E}_1}{{\cal M}_{i-1}}{s''},
      \ldots, \mwrs{\mathsf{E}_r}{{\cal M}_{i-1}}{s''}$ are total due to the fact that $s''$ is total and ${\cal M}_{i-1}$
      assigns total values (by the outer induction hypothesis).

\item $\mathsf{L} =\,\sim\! \mathsf{q}\, \mathsf{E}_1 \cdots \mathsf{E}_r$. Then:
      $$\lsem \mathsf{L}\rsem^{*}_{s''} ({\cal T}^j_{\mathsf{P}_{i},{\cal M}_{i-1}}(\perp),{\cal M}_{i-1}) =
      {\cal M}_{i-1}(\mathsf{q}) \, \mwrs{\mathsf{E}_1}{{\cal M}_{i-1}}{s''}\cdots \mwrs{\mathsf{E}_r}{{\cal M}_{i-1}}{s''}$$
      which is total by the outer induction hypothesis and the fact that all the $\mwrs{\mathsf{E}_1}{{\cal M}_{i-1}}{s''},
      \ldots, \mwrs{\mathsf{E}_r}{{\cal M}_{i-1}}{s''}$ are total due to the fact that $s''$ is total and ${\cal M}_{i-1}$
      assigns total values (by the outer induction hypothesis).

\item The cases $\mathsf{L} = \mathsf{Q} \,\mathsf{E}_1 \cdots \mathsf{E}_r$ and $\mathsf{L} = \,\sim \!\mathsf{Q} \,\mathsf{E}_1 \cdots \mathsf{E}_r$ are very similar to the previous ones, using the fact that $s''$ is total.
\end{itemize}
This completes the proof of the lemma.
\end{proof}
\begin{retheorem}{is_unique}
%
%
Let $\mathsf{P}$ be a stratified $\HOL$ program. Then, $\mathsf{P}$ has a unique equilibrium model.
\end{retheorem}
\begin{proof}
The proof is by induction on $i$. The basis case for $i=1$ follows from Lemma~\ref{sp_uniqueness}.
We prove the result for ${\cal M}_i$, $i>1$, assuming the result
holds for ${\cal M}_{i-1}$. For the sake of contradiction, assume there exists another equilibrium model ${\cal N}$
of $\mathsf{P}_1 \cup \cdots \cup \mathsf{P}_{i}$. Since ${\cal M}_{i-1}$ is the unique equilibrium model
of $\mathsf{P}_1 \cup \cdots \cup \mathsf{P}_{i-1}$, by Lemma~\ref{splitting_submodel}, ${\cal N}$ agrees with ${\cal M}_{i-1}$, and therefore
also with ${\cal M}_i$, for all predicate constant symbols defined in $\mathsf{P}_1 \cup \cdots \cup \mathsf{P}_{i-1}$.
Moreover, since by Lemma~\ref{minimality} ${\cal M}_{i}$ is the least model of $\mathsf{P}_1 \cup \cdots \cup \mathsf{P}_{i}$
that agrees with ${\cal M}_{i-1}$, it follows that ${\cal M}_{i} < {\cal N}$. We construct an interpretation ${\cal N}^*$
of $\mathsf{P}_1 \cup \cdots \cup \mathsf{P}_{i}$, as follows:
$${\cal N}^{*}(\mathsf{p})\, d_1 \cdots d_n =
   \begin{cases}
       T^*,  & \text{if ${\cal M}_{i}(\mathsf{p})\, d_1 \cdots d_n \neq {\cal N}(\mathsf{p})\, d_1 \cdots d_n$}\\
       {\cal N}(\mathsf{p})\, d_1 \cdots d_n, & \text{otherwise}
   \end{cases}$$
By a simple case analysis it follows that ${\cal N}^*$ is $\preceq$-monotonic.
Moreover, it is also easy to verify that ${\cal M}_i \leq {\cal N}^*$ and ${\cal N}^* \preceq {\cal N}$. We show that ${\cal N}^* \neq {\cal N}$.
Since we assumed that ${\cal N}$ is an equilibrium model of
$\mathsf{P}_1 \cup \cdots \cup \mathsf{P}_{i}$, it must be that ${\cal M}_i \nprec{\cal N}$. Since ${\cal M}_i < {\cal N}$ there must
be a predicate constant symbol $\mathsf{p}$ and inputs $d_1, \ldots, d_n$ of appropriate types such that
${\cal M}_i(\mathsf{p})\, d_1 \cdots d_n = F$ and $ {\cal N}(\mathsf{p})\, d_1 \cdots d_n \in \{T^*, T\}$.
Because ${\cal N}$ is total there exist inputs $e_1, \ldots, e_n$ with $d_i \preceq e_i$
such that ${\cal M}_i(\mathsf{p})\, e_1 \cdots e_n = F$ and ${\cal N}(\mathsf{p})\, e_1 \cdots e_n = T$.
But then by the construction of ${\cal N}^*$  we have ${\cal N}^*(\mathsf{p})\, e_1 \cdots e_n = T^* \neq {\cal N}(\mathsf{p})\, e_1 \cdots e_n$
and so we prove that ${\cal N}^* \prec {\cal N}$.

We show that ${\cal N}^*$ is a model of $\mathsf{P}_1 \cup \cdots \cup \mathsf{P}_{i}$,
which will contradict our assumption that ${\cal N}$ is an equilibrium model of $\mathsf{P}_1 \cup \cdots \cup \mathsf{P}_{i}$.

Let $\mathsf{p} \, \mathsf{R}_1\cdots \mathsf{R}_n \leftarrow \mathsf{L}_1,\ldots,\mathsf{L}_m$ be
an arbitrary rule of $\mathsf{P}_1 \cup \cdots \cup \mathsf{P}_{i}$, let $s$ be an arbitrary state, let $d_1,\ldots,d_n$
be elements of the appropriate types, and $s' = s[\mathsf{R}_1/d_1,\ldots,\mathsf{R}_n/d_n]$.
We distinguish cases based on the value of ${\cal N}^*(\mathsf{p})\,d_1\ldots d_n$:

\vspace{0.2cm}
\noindent
{\em Case 1:} ${\cal N}^*(\mathsf{p})\,d_1\ldots d_n = T$. In this case ${\cal N}^*$ satisfies trivially the
given rule.

\vspace{0.2cm}
\noindent
{\em Case 2:} ${\cal N}^*(\mathsf{p})\,d_1\ldots d_n = F$. But then, ${\cal N}(\mathsf{p})\,d_1\ldots d_n = F$
and therefore there exists a literal $\mathsf{L}_i$ in the body of the rule such that
$\lsem \mathsf{L}_i \rsem_{s'}({\cal N}) = F$. We distinguish cases:
\begin{itemize}
%

\item $\mathsf{L}_i = \mathsf{q}\, \mathsf{E}_1 \cdots \mathsf{E}_r$. Notice now that
      $\lsem \mathsf{E}_j\rsem_{s'}({\cal N}) = \lsem \mathsf{E}_j\rsem_{s'}({\cal M}_{i})=
      \lsem \mathsf{E}_j\rsem_{s'}({\cal N}^*)$
      because the $\mathsf{E}_j's$ contain only predicate variables (whose denotations depend on $s'$)
      and predicate constants of lower strata (whose denotations depend only on ${\cal M}_{i-1}$,
      on which ${\cal M}_i$, ${\cal N}$ and therefore ${\cal N}^*$ agree).
      Moreover, since
      ${\cal N}(\mathsf{q})\, \lsem \mathsf{E}_1\rsem_{s'}({\cal N}) \cdots \lsem \mathsf{E}_r\rsem_{s'}({\cal N}) = F$
      and ${\cal M}_{i} < {\cal N}$, we get that
      ${\cal M}_{i}(\mathsf{q})\, \lsem \mathsf{E}_1\rsem_{s'}({\cal M}_{i}) \cdots
      \lsem \mathsf{E}_r\rsem_{s'}({\cal M}_{i}) = F$. Therefore, by the definition of
      ${\cal N}^*$, we get that ${\cal N}^*(\mathsf{q})\, \lsem \mathsf{E}_1\rsem_{s'}({\cal N}^*) \cdots \lsem \mathsf{E}_r\rsem_{s'}({\cal N}^*) = F$.

\item $\mathsf{L}_i =\,\sim\! \mathsf{q}\, \mathsf{E}_1 \cdots \mathsf{E}_r$, $\mathsf{L}_i = \mathsf{Q} \,\mathsf{E}_1
      \cdots \mathsf{E}_r$, and $\mathsf{L}_i =\,\sim\! \mathsf{Q} \,\mathsf{E}_1 \cdots \mathsf{E}_r$. The proofs for
      these cases are similar to that of the previous case.

\end{itemize}
Therefore, ${\cal N}^*$ satisfies the rule in this case.

\vspace{0.2cm}
\noindent
{\em Case 3:} ${\cal N}^*(\mathsf{p})\,d_1\ldots d_n = T^*$. But then ${\cal M}_i(\mathsf{p})\,d_1\ldots d_n \leq T^*$
We claim that there exists a literal $\mathsf{L}_i$ in the body of the rule such that
$\lsem \mathsf{L}_i \rsem_{s'}({\cal N}^*) \leq T^*$. Assume, for the sake of contradiction, that
for every $\mathsf{L}_i$, $\lsem \mathsf{L}_i \rsem_{s'}({\cal N}^*) = T$. We show that then
it will also be the case that for every $\mathsf{L}_i$, $\lsem \mathsf{L}_i \rsem_{s'}({\cal M}_i) = T$. This is a contradiction
on the fact that ${\cal M}_i$ is a model. We distinguish cases:

\begin{itemize}
%

\item $\mathsf{L}_i = \mathsf{q}\, \mathsf{E}_1 \cdots \mathsf{E}_r$. Notice now that
      $\lsem \mathsf{E}_j\rsem_{s'}({\cal N}) = \lsem \mathsf{E}_j\rsem_{s'}({\cal M}_{i})=
      \lsem \mathsf{E}_j\rsem_{s'}({\cal N}^*)$
      because the $\mathsf{E}_j's$ contain only predicate variables (whose denotations depend on $s'$)
      and predicate constants of lower strata (whose denotations depend only on ${\cal M}_{i-1}$,
      on which ${\cal M}_i$, ${\cal N}$ and therefore ${\cal N}^*$ agree). Moreover, since
      ${\cal N}^*(\mathsf{q})\, \lsem \mathsf{E}_1\rsem_{s'}({\cal N}^*) \cdots \lsem \mathsf{E}_r\rsem_{s'}({\cal N}^*) = T$,
      it has to be the case that ${\cal M}_i(\mathsf{q})\, \lsem \mathsf{E}_1\rsem_{s'}({\cal N}^*) \cdots
      \lsem \mathsf{E}_r\rsem_{s'}({\cal N}^*) = T$. Indeed if otherwise, since it also is
      ${\cal N}(\mathsf{q})\, \lsem \mathsf{E}_1\rsem_{s'}({\cal N}^*) \cdots
      \lsem \mathsf{E}_r\rsem_{s'}({\cal N}^*) = T$  we have that ${\cal N}(\mathsf{q})\, \lsem \mathsf{E}_1\rsem_{s'}({\cal N}^*) \cdots
      \lsem \mathsf{E}_r\rsem_{s'}({\cal N}^*) \neq {\cal M}_i(\mathsf{q})\, \lsem \mathsf{E}_1\rsem_{s'}({\cal N}^*) \cdots
      \lsem \mathsf{E}_r\rsem_{s'}({\cal N}^*) $ and by construction it should have been ${\cal N}^*(\mathsf{q})\, \lsem \mathsf{E}_1\rsem_{s'}({\cal N}^*) \cdots
      \lsem \mathsf{E}_r\rsem_{s'}({\cal N}^*)= T^*$. We conclude that
      ${\cal M}_i(\mathsf{q})\, \lsem \mathsf{E}_1\rsem_{s'}({\cal N}^*) \cdots
      \lsem \mathsf{E}_r\rsem_{s'}({\cal N}^*) = {\cal M}_i(\mathsf{q})\, \lsem \mathsf{E}_1\rsem_{s'}({\cal M}_i) \cdots
      \lsem \mathsf{E}_r\rsem_{s'}({\cal M}_i) = T$.

\item $\mathsf{L}_i =\,\sim\! \mathsf{q}\, \mathsf{E}_1 \cdots \mathsf{E}_r$, $\mathsf{L}_i = \mathsf{Q} \,\mathsf{E}_1
      \cdots \mathsf{E}_r$, and $\mathsf{L}_i =\,\sim\! \mathsf{Q} \,\mathsf{E}_1 \cdots \mathsf{E}_r$. The proofs for
      these cases are similar to that of the previous case.
\end{itemize}
The above implies that ${\cal M}_i$ does not satisfy the given rule, which is a contradiction.
Therefore, there exists $\mathsf{L}_i$ such that $\lsem \mathsf{L}_i \rsem_{s'}({\cal N}^*) \leq T^*$
and consequently ${\cal N}^*$ satisfies the given rule.
\end{proof}

\section{Proofs of Section~\ref{sec:representation-theorem}}
\newcommand{\order}{width\xspace}

In this section we give detailed proofs of the following two theorems, stated in Section 4:

\begin{retheorem}{definability_of_exact}
For each type $\pi$ and every total $d\in\mo{\pi}$, there exists a stratified program $\mathsf{P}_d$ with a constant $\ce{d} : \pi$ such that
$\mathsf{P}_d$ has a unique equilibrium model $\mathcal{M}$ with $\mathcal{M}(\ce{d}) = d$.
\end{retheorem}

\begin{retheorem}{definability_of_non_exact}
For each type $\pi$ and $d\in\mo{\pi}$, there exists a stratified program $\mathsf{P}_d$ with a constant $\cst{d} : o \rightarrow \pi$ such that
$\mathsf{P}_d$ has a unique equilibrium model $\mathcal{M}$ with
\[
\mathcal{M}(\cst{d})(u)=\begin{cases}
  d, & \text{ if } u = \NF\\
  \tv{d}, & \text{ if } u = \T\\
  \lambda \overline{x}.\F, & \text{ if } u = \F
\end{cases}\]
\end{retheorem}

Firstly, the following definition introduces a way to order types.

\begin{definition}
  We define for any type $\pi$ the {\emph \order} of $\pi$ as
  $\order(o) = 1$ and for any type $\pi \neq o$ such that $\pi = \pi_1\rightarrow\cdots\rightarrow\pi_n\rightarrow o$,
  $\order(\pi)=\order(\pi_1)+\cdots+\order(\pi_n)+1$.
\end{definition}

Let $d \in \mo{\tau}$ be a element of our domain for some type $\tau$.
In the rest of this section we are going to define a program $\mathsf{P}_{\tau}$ containing
predicate constants $\ce{d}$ and $\cst{d}$ for all domain elements $d' \in \tau'$ for all types $\tau'$ of smaller \order than $\tau$.
We will show that $\mathsf{P}_{\tau}$ has a unique equilibrium model which assigns the desirable meaning to the constants $\ce{d}$ and $\cst{d}$.

We start with the following technical lemmata about the $collapse$ function:

\begin{lemma}\label{tv_preceq}
Let $\pi$ be a type and let $d \in \mo{\pi}$.
Then, $d \preceq \tv{d}$.
\end{lemma}
\begin{proof}
By induction on the structure of $\pi$.
If $\pi = o$, it follows from the definition.
Suppose that $\pi = \rho\rightarrow \pi'$ and the lemma holds for $\pi'$.
Let $d' \in \mo{\rho}$. By the induction hypothesis, we have $d(d') \preceq \tv{d(d')} = \tv{d}(d')$.
Since this holds for all $d' \in \mo{\rho}$, we have $d \preceq \tv{d}$.
\end{proof}

\begin{lemma}\label{collapse_total}
Let $\pi$ be a type and let $d \in \mo{\pi}$. Then, $\tv{d}$ is total.
\end{lemma}
\begin{proof}
Trivial, by induction on the structure of $\pi$.
\end{proof}

\begin{corollary}\label{exists_total}
Let $\pi$ be a type and let $d \in \mo{\pi}$. Then, there exists a total $e \in \mo{\pi}$ such that $d \preceq e$.
\end{corollary}
\begin{proof}
Follows from Lemma~\ref{tv_preceq} and Lemma~\ref{collapse_total}.
\end{proof}

\begin{lemma}\label{preceq_tv}
Let $\pi$ be a type and let $d_1,d_2 \in \mo{\pi}$.
If $d_1 \preceq d_2$, then $\tv{d_1} = \tv{d_2}$.
\end{lemma}
\begin{proof}
By induction on the structure of $\pi$.
If $\pi = o$, it follows from the definitions.
Suppose that $\pi = \rho\rightarrow \pi'$ and the lemma holds for $\pi'$.
Let $d \in \mo{\rho}$. Since $d_1 \preceq d_2$, we have $d_1(d) \preceq d_2(d)$.
By the induction hypothesis, we have $\tv{d_1(d)} = \tv{d_2(d)}$, so that $\tv{d_1}(d) = \tv{d_2}(d)$.
Since this holds for all $d \in \mo{\rho}$, we have $\tv{d_1} = \tv{d_2}$.
\end{proof}

For each type $\pi$ and $d \in \mo{\pi}$ we introduce a predicate constant $\cst{d}: o\rightarrow\pi$.
Also, we introduce that predicate constants $\mathsf{neg} : o \rightarrow o$ and $\mathsf{neg\_neg} : o \rightarrow o$.

Let $d \in \mo{\pi}$ for some type $\pi$.
We define an expression $\Est{d,\mathsf{V},\mathsf{U}} : o$  for each variable $\mathsf{V} : \pi$ and expression $\mathsf{U} : o$
as:
\[
\begin{array}{rl}
\Est{d,\mathsf{V},\mathsf{U}} = &\bigwedge\{ \mathsf{V}~(\cst{d_1}~\mathsf{U})~\cdots~(\cst{d_n}~\mathsf{U}) \mid d~d_1~\cdots~d_n = \T\}\wedge
\\& \bigwedge\{ \mathsf{neg\_neg} (\mathsf{V}~(\cst{d_1}~\mathsf{U})~\cdots~(\cst{d_n}~\mathsf{U})) \mid d~d_1\cdots~d_n = \NF\}\wedge
\\& \bigwedge\{ \mathsf{neg} (\mathsf{V}~(\cst{d_1}~\mathsf{U})~\cdots~(\cst{d_n}~\mathsf{U})) \mid d~d_1~\cdots~d_n = \F\}
\end{array}
\]

The intuitive meaning of $\Est{d,\mathsf{V},\mathsf{U}}$ is that this expression is true under an interpretation
$\mathcal{I}$ and a state $s$ when $s$ assigns the value $d$ to the variable $\mathsf{V}$.
As explained in section~\ref{sec:representation-theorem}, in order to define non-total relations, we need to define somehow the value $\NF$.
This is the point of the expression $\mathsf{U}$.
It is an expression that it is assumed to have the value $\NF$ under the interpretation $\mathcal{I}$.

Let $\mathcal{M}^*$ be an interpretation such that:
\begin{itemize}
  \item For all $x \in \mo{o}$, $\mathcal{M}^*(\mathsf{neg})(x) = \pnot x$.
  \item For all $x \in \mo{o}$, $\mathcal{M}^*(\mathsf{neg\_neg})(x) = \tv{x}$.
  \item For all types $\pi$ and $d \in \mo{\pi}$,
  \[\mathcal{M}^*(\cst{d})(u) = \begin{cases}
  d, & \text{ if } u = \NF\\
  \tv{d}, & \text{ if } u = \T\\
  \lambda \overline{x}.\F, & \text{ if } u = \F
\end{cases}\]
\end{itemize}

The following definition is just a technical notation that we use in the following proofs.

\begin{definition}
  Let $\pi$ be a type such that $\pi = \pi_1\rightarrow\cdots\rightarrow\pi_n\rightarrow o$.
  We define $arg(\pi) = \mo{\pi_1}\times\cdots\times\mo{\pi_n}$.
\end{definition}

Now we can state formally some lemmata about the meaning of $\Est{d,\mathsf{V},\mathsf{U}}$ expressions.

\begin{lemma}\label{lemma_E_U_tst}
Let $\mathcal{I}$ be an interpretation of $\mathsf{P}^*_{\tau}$, $\pi$ be a type and $s$ be a state.
For any variables $\mathsf{V} : \pi$, $\mathsf{U} : o$ and any $d\in\mo{\pi}$,
if $\mathcal{I}(\cst{d'})=\mathcal{M}^*(\cst{d'})$ for any $d'$ of type of smaller \order than $\pi$, $s(\mathsf{V}) = d$ and $s(\mathsf{U}) = \NF$,
then $\mwrs{\Est{d,\mathsf{V},\mathsf{U}}}{\mathcal{I}}{s} = \T$.
\end{lemma}
\begin{proof}
Trivial.
\end{proof}

\begin{lemma}\label{lemma_E_U_t}
Let $\mathcal{I}$ be an interpretation of $\mathsf{P}^*_{\tau}$, $\pi$ be a type and $s$ be a state. For any variables $\mathsf{V} : \pi$, $\mathsf{U} : o$ and any $d\in\mo{\pi}$,
if $\mathcal{I}(\cst{d'})=\mathcal{M}^*(\cst{d'})$ for any $d'$ of type of smaller \order than $\pi$, $s(\mathsf{V}) = d$ and $s(\mathsf{U}) = \T$, then $\mwrs{\Est{d,\mathsf{V},\mathsf{U}}}{\mathcal{I}}{s} = \T$.
\end{lemma}
\begin{proof}
By previous lemma and $\preceq$-monotonicity.
\end{proof}

\begin{lemma}\label{lemma_U_tst_E_t}
Let $\pi$ be a type and $s$ be a state. For any variables $\mathsf{V} : \pi$, $\mathsf{U} : o$ and any $d\in\mo{\pi}$,
if $\mwrs{\Est{d,\mathsf{V},\mathsf{U}}}{\mathcal{M}^*}{s} = \T$ and $s(\mathsf{U}) = \NF$, then $d \preceq s(\mathsf{V})$.
\end{lemma}
\begin{proof}
We are going to show that for each $(d_1,\ldots,d_n)\in arg(\pi)$, we have $d~d_1~\cdots~d_n \preceq s(\mathsf{V})~d_1~\cdots~d_n$.
If $d~d_1~\cdots~d_n = \T$, then $\mathsf{V}~(\cst{d_1} \mathsf{U})~\cdots~(\cst{d_n} \mathsf{U})$ is a conjunct of $\Est{d,\mathsf{V},\mathsf{U}}$.
Since $\mwrs{\Est{d,\mathsf{V},\mathsf{U}}}{\mathcal{M}^*}{s} = \T$, we have that $s(\mathsf{V})~d_1~\cdots~d_n=\mwrs{\mathsf{V}~(\cst{d_1} \mathsf{U})~\cdots~(\cst{d_n} \mathsf{U})}{\mathcal{M}^*}{s} = \T$.
If $d~d_1~\cdots~d_n = \F$, $\mathsf{neg} (\mathsf{V}~(\cst{d_1} \mathsf{U})~\cdots~(\cst{d_n} \mathsf{U}))$ is a conjunct of $\Est{d,\mathsf{V},\mathsf{U}}$.
Since $\mwrs{\Est{d,\mathsf{V},\mathsf{U}}}{\mathcal{M}^*}{s} = \T$, we have that $\mwrs{\mathsf{neg} (\mathsf{V}~(\cst{d_1} \mathsf{U})~\cdots~(\cst{d_n} \mathsf{U}))}{\mathcal{M}^*}{s} = \T$,
so that $s(\mathsf{V})~d_1~\cdots~d_n = \mwrs{\mathsf{V}~(\cst{d_1} \mathsf{U})~\cdots~(\cst{d_n} \mathsf{U})}{\mathcal{M}^*}{s} = \F$.
If $d~d_1~\cdots~d_n = \NF$, then $\mathsf{neg\_neg} (\mathsf{V}~(\cst{d_1} \mathsf{U})~\cdots~(\cst{d_n} \mathsf{U}))$ is a conjunct of $\Est{d,\mathsf{V},\mathsf{U}}$.
Then $s(\mathsf{V})~d_1~\cdots~d_n$ is either $\T$ or $\NF$. In any case, we have that $d~d_1~\cdots~d_n \preceq s(\mathsf{V})~d_1~\cdots~d_n$.
\end{proof}

\begin{lemma}\label{lemma_U_tst_E_tst}
Let $\pi$ be a type and $s$ be a state. For any variables $\mathsf{V} : \pi$, $\mathsf{U} : o$ and any $d\in\mo{\pi}$,
if $\mwrs{\Est{d,\mathsf{V},\mathsf{U}}}{\mathcal{M}^*}{s} =\NF$ and $s(\mathsf{U}) = \NF$, then $\tv{s(\mathsf{V})} = \tv{d}$.
\end{lemma}
\begin{proof}
Similarly to the proof of the previous lemma.
\end{proof}

\begin{lemma}\label{lemma_U_t}
Let $\pi$ be a type and $s$ be a state. For any variables $\mathsf{V} : \pi$, $\mathsf{U} : o$ and any $d\in\mo{\pi}$,
if $\mwrs{\Est{d,\mathsf{V},\mathsf{U}}}{\mathcal{M}^*}{s}\in\{\T, \NF\}$ and $s(\mathsf{U}) = \T$, then $\tv{s(\mathsf{V})} = \tv{d}$.
\end{lemma}
\begin{proof}
We are going to show that for each $(d_1,\ldots,d_n) \in arg(\pi)$, we have $\tv{d~d_1~\cdots~d_n} = \tv{s(\mathsf{V})~d_1~\cdots~d_n}$.
If $d~d_1~\cdots~d_n = \T$, then $\mathsf{V}~(\cst{d_1} \mathsf{U})~\cdots~(\cst{d_n} \mathsf{U})$ is a conjunct of $\Est{d,\mathsf{V},\mathsf{U}}$.
Since $\mwrs{\Est{d,\mathsf{V},\mathsf{U}}}{\mathcal{M}^*}{s} \in \{\T, \NF\}$, we have that $\mwrs{\mathsf{V}~(\cst{d_1} \mathsf{U})~\cdots~(\cst{d_n} \mathsf{U})}{\mathcal{M}^*}{s} \in \{\T, \NF\}$.
By the definition of $\mathcal{M}^*$, $\mwrs{\mathsf{V}~(\cst{d_1} \mathsf{U})~\cdots~(\cst{d_n} \mathsf{U})}{\mathcal{M}^*}{s} = s(\mathsf{V})~\tv{d_1}~\cdots~\tv{d_n}$.
By the monotonicity of $s(\mathsf{V})$, $s(\mathsf{V})~d_1~\cdots~d_n \in \{\T, \NF\}$. In any case, we have that $\tv{s(\mathsf{V})~d_1~\cdots~d_n} = \T$.
Similar arguments can be used for the cases $d~d_1~\cdots~d_n = \NF$ and $d~d_1~\cdots~d_n = \F$.
\end{proof}

Now, we are ready to define the program that constructs the constants $\cst{d}$ for all domain element $d$.

For any type $\pi$ and each $d \in \mo{\pi}$ we define a program $\mathsf{P}_{\cst{d}}$ consisting of clauses:
\begin{itemize}
  \item
  \[
    \mathsf{neg}~\mathsf{R} \leftarrow \pnot\mathsf{R}
  \]
  \[
    \mathsf{neg\_neg}~\mathsf{R} \leftarrow \mathsf{neg}(\mathsf{neg}~\mathsf{R})
  \]
  \item For each $d_1,\ldots,d_n$ such that $d~d_1~\cdots~d_n = \T$
\[
  \cst{d}~\mathsf{U}~\mathsf{R}_1~\cdots~\mathsf{R}_n \leftarrow \mathsf{neg\_neg} (\mathsf{U})\wedge \Est{d_1,\mathsf{R}_1,\mathsf{U}}\wedge\cdots\wedge\Est{d_n,\mathsf{R}_n,\mathsf{U}}
\]
  \item For each $d_1,\ldots,d_n$ such that $d~d_1~\cdots~d_n = \NF$
\[
  \cst{d}~\mathsf{U}~\mathsf{R}_1~\cdots~\mathsf{R}_n \leftarrow \mathsf{U}\wedge\Est{d_1,\mathsf{R}_1,\mathsf{U}}\wedge\cdots\wedge\Est{d_n,\mathsf{R}_n,\mathsf{U}}
\]
\end{itemize}
$\mathsf{P}^*_{\tau}$ is the union of every $\mathsf{P}_{\cst{d}}$ for all elements $d \in \mo{\pi}$ for every type $\pi$ with $\order(\pi) \leq \order(\tau)$.

Now, we can prove the main result about the $\cst{d}$ constants.
We will firstly show that $\mathcal{M}^*$ is a model of $\mathsf{P}^*_{\tau}$.
Then, we will show that $\mathcal{M}^*$ is $\preceq$-minimal.
Since $\mathcal{M}^*$ is also total, we have that $\mathcal{M}^*$ is an equilibrium model of $\mathsf{P}^*_{\tau}$.
Now, we argue that $\mathsf{P}^*_{\tau}$ is stratified. In order to see that notice that for every clause
$\cst{d}\ \overline{\mathsf{R}} \lrule \mathsf{B}$, for every predicate constant $\cst{d'}$ appearing in $\mathsf{B}$,
the type of $\cst{d'}$ is of smaller \order than the type of $\cst{d}$.
Therefore, by Theorem~\ref{is_unique}, $\mathcal{M}^*$ is the unique equilibrium model of $\mathsf{P}^*_{\tau}$.

\begin{lemma}
$\mathcal{M}^*$ is a model of $\mathsf{P}^*_{\tau}$.
\end{lemma}
\begin{proof}
Let $s$ be a state and $C$ be a clause of $\mathsf{P}^*_{\tau}$.
If $C$ is $\mathsf{neg}~\mathsf{R} \leftarrow \pnot\mathsf{R}$ or $\mathsf{neg\_neg}~\mathsf{R} \leftarrow \mathsf{neg}(\mathsf{neg}~\mathsf{R})$,
it is easy to see that $\mathcal{M}^*$ satisfies $C$.
Then, suppose that $C$ is of the form
$\cst{d}~\mathsf{U}~\mathsf{R}_1~\cdots~\mathsf{R}_n \leftarrow \mathsf{neg\_neg} (\mathsf{U})\wedge \Est{d_1,\mathsf{R}_1,\mathsf{U}}\wedge\cdots\wedge\Est{d_n,\mathsf{R}_n,\mathsf{U}}$
for some $d \in \mo{\pi}, (d_1,\ldots,d_n)\in arg(\pi)$ such that $d~d_1~\cdots~d_n = \T$.
If $s(\mathsf{U}) = \F$ or $\mwrs{\Est{d_i,\mathsf{R}_i,\mathsf{U}}}{\mathcal{M}^*}{s} = \F$ for some $i$, we are done.
If $s(\mathsf{U}) = \T$ and $\mwrs{\Est{d_i,\mathsf{R}_i,\mathsf{U}}}{\mathcal{M}^*}{s} \in \{\T, \NF\}$ for all $i \in \{1,\ldots,n\}$, then,
by Lemma~\ref{lemma_U_t}, we have that $\tv{s(\mathsf{R}_i)} = \tv{d_i}$ for all $i \in \{1,\ldots,n\}$.
We have $\mwrs{\cst{d}~\mathsf{U}~\mathsf{R}_1~\cdots~\mathsf{R}_n}{\mathcal{M}^*}{s} =
\mathcal{M}^*(\cst{d})~s(\mathsf{U})~s(\mathsf{R}_1)~\cdots~s(\mathsf{R}_n) =
\mathcal{M}^*(\cst{d})~\T~\tv{d_1}~\cdots~\tv{d_n} =
\tv{d~d_1~\cdots~d_n} = \T$.
If $s(\mathsf{U}) = \NF$ and $\mwrs{\Est{d_i,\mathsf{R}_i,\mathsf{U}}}{\mathcal{M}^*}{s} \in \{\T, \NF\}$ for all $i \in \{1,\ldots,n\}$
and $\mwrs{\Est{d_i,\mathsf{R}_i,\mathsf{U}}}{\mathcal{M}^*}{s} = \NF$ for some $i$, then by Lemma~\ref{lemma_U_tst_E_t}, Lemma~\ref{lemma_U_tst_E_tst}, and Lemma~\ref{preceq_tv}, we have that
$\tv{s(\mathsf{R}_i)} = \tv{d_i}$ for all $i \in \{1,\ldots,n\}$.
We have $\mwrs{\cst{d}~\mathsf{U}~\mathsf{R}_1~\cdots~\mathsf{R}_n}{\mathcal{M}^*}{s} = d~s(\mathsf{R}_1)~\cdots~s(\mathsf{R}_n)$.
By Lemma~\ref{tv_preceq} and the monotonicity of $d$, $d~s(\mathsf{R}_1)~\cdots~s(\mathsf{R}_n) \preceq d~\tv{s(\mathsf{R}_1)}~\cdots~\tv{s(\mathsf{R}_n)}$ and
$d~d_1~\cdots~d_n \preceq d~\tv{d_1}~\cdots~\tv{d_n}$.
Therefore, $\mwrs{\cst{d}~\mathsf{U}~\mathsf{R}_1~\cdots~\mathsf{R}_n}{\mathcal{M}^*}{s} \in \{\T, \NF\}$.
If $s(\mathsf{U}) = \NF$ and $\mwrs{\Est{d_1,\mathsf{R}_1,\mathsf{U}}}{\mathcal{M}^*}{s} = \T$ for all $i \in \{1,\ldots,n\}$, then,
by Lemma~\ref{lemma_U_tst_E_t}, we have that $d_i \preceq s(\mathsf{R}_i)$ for all $i \in \{1,\ldots,n\}$.
By the monotonicity of $d$, $d~d_1~\cdots~d_n \preceq d~s(\mathsf{R}_1)~\cdots~s(\mathsf{R}_n) = \mwrs{\cst{d}~\mathsf{U}~\mathsf{R}_1~\cdots~\mathsf{R}_n}{\mathcal{M}^*}{s}$.
Therefore, $\mwrs{\cst{d}~\mathsf{U}~\mathsf{R}_1~\cdots~\mathsf{R}_n}{\mathcal{M}^*}{s} = \T$.

Suppose now that $C$ is of the form
$\cst{d}~\mathsf{U}~\mathsf{R}_1~\cdots~\mathsf{R}_n \leftarrow \mathsf{U}\wedge\Est{d_1,\mathsf{R}_1,\mathsf{U}}\wedge\cdots\wedge\Est{d_n,\mathsf{R}_n,\mathsf{U}}$
for some $d \in \mo{\pi}, (d_1,\ldots,d_n)\in arg(\pi)$ such that $d~d_1~\cdots~d_n = \NF$.
If $s(\mathsf{U}) = \F$ or $\mwrs{\Est{d_1,\mathsf{R}_1,\mathsf{U}}}{\mathcal{M}^*}{s} = \F$ for some $i$, we are done.
If $s(\mathsf{U}) = \T$ and $\mwrs{\Est{d_1,\mathsf{R}_1,\mathsf{U}}}{\mathcal{M}^*}{s} \in \{\T, \NF\}$ for all $i \in \{1,\ldots,n\}$, then,
by Lemma~\ref{lemma_U_t}, we have that $\tv{s(\mathsf{R}_i)} = \tv{d_i}$ for all $i \in \{1,\ldots,n\}$.
We have $\mwrs{\cst{d}~\mathsf{U}~\mathsf{R}_1~\cdots~\mathsf{R}_n}{\mathcal{M}^*}{s} =
\mathcal{M}^*(\cst{d})~s(\mathsf{U})~s(\mathsf{R}_1)~\cdots~s(\mathsf{R}_n) =
\mathcal{M}^*(\cst{d})~\T~\tv{d_1}~\cdots~\tv{d_n} =
\tv{d~d_1~\cdots~d_n} = \T$.
If $s(\mathsf{U}) = \NF$ and $\mwrs{\Est{d_1,\mathsf{R}_1,\mathsf{U}}}{\mathcal{M}^*}{s} \in \{\T, \NF\}$ for all $i \in \{1,\ldots,n\}$,
then, by Lemma~\ref{lemma_U_tst_E_t}, Lemma~\ref{lemma_U_tst_E_tst} and Lemma~\ref{preceq_tv}, we have that
$\tv{s(\mathsf{R}_i)} = \tv{d_i}$ for all $i \in \{1,\ldots,n\}$.
We have $\mwrs{\cst{d}~\mathsf{U}~\mathsf{R}_1~\cdots~\mathsf{R}_n}{\mathcal{M}^*}{s} = d~s(\mathsf{R}_1)~\cdots~s(\mathsf{R}_n)$.
By Lemma~\ref{tv_preceq} and the monotonicity of $d$, $d~s(\mathsf{R}_1)~\cdots~s(\mathsf{R}_n) \preceq d~\tv{s(\mathsf{R}_1)}~\cdots~\tv{s(\mathsf{R}_n)}$ and
$d~d_1~\cdots~d_n \preceq d~\tv{d_1}~\cdots~\tv{d_n}$.
Therefore, $\mwrs{\cst{d}~\mathsf{U}~\mathsf{R}_1~\cdots~\mathsf{R}_n}{\mathcal{M}^*}{s} \in \{\T, \NF\}$.
Since $\mwrs{\mathsf{U}}{\mathcal{M}^*}{s} = \NF$, $C$ is satisfied.
\end{proof}

\begin{lemma}
$\mathcal{M}^*$ is a $\preceq$-minimal model of $\mathsf{P}^*_{\tau}$.
\end{lemma}
\begin{proof}
Suppose, for the sake of contradiction, that there exists a model $\mathcal{N}$ of $\mathsf{P}^*_{\tau}$ such that $\mathcal{N} \prec \mathcal{M}^*$.
Then there exist some $d \in \mo{\pi}$ for some type $\pi$, some $u \in \mo{o}$ and some $(d_1,\ldots,d_n)\in arg(\pi)$ such that
$\mathcal{M}^*(\cst{d})~u~d_1~\cdots~d_n=\T$ and $\mathcal{N}(\cst{d})~u~d_1~\cdots~d_n=\NF$.
We assume that $\pi$ is of minimal \order, that is $\mathcal{N}(\cst{d'}) = \mathcal{M}^*(\cst{d'})$ for all $d'\in\pi'$ such that $\pi'$ has lower \order than $\pi$.
If $u=\F$, then the assumption $\mathcal{M}^*(\cst{d})~u~d_1~\cdots~d_n=\T$ can't be true.
Therefore, $u\in\{\T,\NF\}$.
Also, by the definition of $\mathcal{M}^*$, we have that $d~d_1~\cdots~d_n \in \{\T,\NF\}$.
Firstly suppose that $d~d_1~\cdots~d_n = \T$.
Then, there exists the clause $C = \cst{d}~\mathsf{U}~\mathsf{R}_1~\cdots~\mathsf{R}_n \leftarrow \mathsf{neg\_neg} (\mathsf{U})\wedge \Est{d_1,\mathsf{R}_1,\mathsf{U}}\wedge\cdots\wedge\Est{d_n,\mathsf{R}_n,\mathsf{U}}$ in $\mathsf{P}^*_{\tau}$.
Let $s$ be a state such that $s(\mathsf{U}) = u$ and $s(\mathsf{R}_i) = d_i$ for all $i \in \{1,\ldots,n\}$.
We have $\mwrs{\mathsf{neg\_neg} (\mathsf{U})}{\mathcal{N}}{s} = \T$ and by Lemmata \ref{lemma_E_U_tst} and \ref{lemma_E_U_t},
$\mwrs{\Est{d_i,\mathsf{R}_i,\mathsf{U}}}{\mathcal{N}}{s} = \T$ for all $i \in \{1,\ldots,n\}$.
Since $\mathcal{N}$ is a model of $\mathsf{P}^*_{\tau}$,
$\mwrs{\cst{d}~\mathsf{U}~\mathsf{R}_1~\cdots~\mathsf{R}_n}{\mathcal{N}}{s} = \mathcal{N}(\cst{d})~u~d_1~\cdots~d_n$ must be $\T$,
which is a contradiction.

Then, suppose that $d~d_1~\cdots~d_n = \NF$.
Then, there exists the clause $C = \cst{d}~\mathsf{U}~\mathsf{R}_1~\cdots~\mathsf{R}_n \leftarrow \mathsf{U}\wedge\Est{d_1,\mathsf{R}_1,\mathsf{U}}\wedge\cdots\wedge\Est{d_n,\mathsf{R}_n,\mathsf{U}}$.
By the definition of $\mathcal{M}^*$, since $\mathcal{M}^*(\cst{d})~u~d_1~\cdots~d_n=\T$, $u$ must be $\T$.
We have $\mwrs{\mathsf{U}}{\mathcal{N}}{s} = \T$ and by Lemma~\ref{lemma_E_U_t}, $\mwrs{\Est{d_i,\mathsf{R}_i,\mathsf{U}}}{\mathcal{N}}{s} = \T$ for all $i \in \{1,\ldots,n\}$.
Since $\mathcal{N}$ is a model of $\mathsf{P}^*_{\tau}$,
$\mwrs{\cst{d}~\mathsf{U}~\mathsf{R}_1~\cdots~\mathsf{R}_n}{\mathcal{N}}{s} = \mathcal{N}(\cst{d})~u~d_1~\cdots~d_n$ must be $\T$,
which is a contradiction.
\end{proof}

This is how we construct the constants $\cst{d}$ for each element $d$.
Now we are going to expand $\mathsf{P}^*_{\tau}$ into a program $\mathsf{P}_{\tau}$ which also defines the constants $\ce{d}$.

For any type $\pi$ and each total $e \in \mo{\pi}$ we introduce a predicate constant $\ce{e}: \pi$.

Let $e \in \mo{\pi}$ for some type $\pi$ such that $e$ is total.
We define an expression $\Ee{e,\mathsf{V}} : o$  for each variable $\mathsf{V} : \pi$ as:
\[
\begin{array}{rl}
\Ee{e,\mathsf{V}} = &\bigwedge\{\mathsf{neg\_neg}(\mathsf{V}~\ce{d_1} \cdots \ce{d_n}) \mid e~d_1~\cdots~d_n = \T, d_1,\ldots,d_n \text{ are total}\}\wedge
\\& \bigwedge\{ \mathsf{neg} (\mathsf{V}~\ce{d_1} \cdots \ce{d_n}) \mid e~d_1~\cdots~d_n = \F, d_1,\ldots,d_n \text{ are total}\}
\end{array}
\]

The intuition behind the definition of $\Ee{e,\mathsf{V}}$ is similar to the one of $\Est{d,\mathsf{V},\mathsf{U}}$.
That is, $\Ee{e,\mathsf{V}}$ is true under an interpretation $\mathcal{I}$ and a state $s$ when $s$ assigns the value $e$ to the variable $\mathsf{V}$.
However, $\Ee{e,\mathsf{V}}$ is only defined for total elements $e$. Therefore we don't need to provide an expression $\mathsf{U}$ for $\NF$.

Let $\mathcal{M}$ be the interpretation such that is the least superset of $\mathcal{M}^*$ such that for any total $e \in \mo{\pi}$
$\mathcal{M}(\ce{e}) = e$.

Now we state some lemmata that formalize the intuition behind the expressions $\Ee{e,\mathsf{V}}$.

\begin{lemma}\label{lemma_E}
Let $\mathcal{I}$ be an interpretation of $\mathsf{P}_{\tau}$, $\pi$ be a type and $s$ be a state. For any variable $\mathsf{V} : \pi$ and any total $e\in\mo{\pi}$,
if $\mathcal{I}(\ce{e'})=\mathcal{M}(\ce{e'})$ for any $e'$ of type of smaller \order than $\pi$ and $s(\mathsf{V}) = e$, then $\mwrs{\mathsf{E}_{e,\mathsf{V}}}{\mathcal{I}}{s} = \T$.
\end{lemma}
\begin{proof}
Trivial.
\end{proof}

\begin{lemma}\label{lemma_E_2}
Let $\pi$ be a type and $s$ be a state. For any variables $\mathsf{V} : \pi$ and any $e\in\mo{\pi}$,
if $\mwrs{\mathsf{E}_{e,\mathsf{V}}}{\mathcal{M}}{s} \in  \{\T, \NF\}$, then $\tv{s(\mathsf{V})} = \tv{e}$.
\end{lemma}
\begin{proof}
We are going to show that for any $(d_1,\ldots,d_n) \in arg(\pi)$, we have $\tv{e~d_1~\cdots~d_n} = \tv{s(\mathsf{V})~d_1~\cdots~d_n}$.
By Lemma~\ref{exists_total}, there exist some $e_1$,\ldots,$e_n$ such that $e_1,\ldots,e_n$ are total and $d_i \preceq e_i$ for all $i \in \{1,\ldots,n\}$.
By the monotonicity of $e$ and $s(\mathsf{V})$, we have that $e~d_1~\cdots~d_n \preceq e~e_1~\cdots~e_n$ and $s(\mathsf{V})~d_1~\cdots~d_n \preceq s(\mathsf{V})~e_1~\cdots~e_n$.
By Lemma~\ref{preceq_tv}, we have $\tv{e~e_1~\cdots~e_n} = \tv{e~d_1~\cdots~d_n}$ and $\tv{s(\mathsf{V})~e_1~\cdots~e_n} = \tv{s(\mathsf{V})~d_1~\cdots~d_n}$.
Therefore, it suffices to show that $\tv{e~e_1~\cdots~e_n} = \tv{s(\mathsf{V})~e_1~\cdots~e_n}$.
Since $e, e_1,\ldots,e_n$ are total, $e~e_1~\cdots~e_n$ can be either $\T$ or $\F$.
If $e~e_1~\cdots~e_n = \T$, then $\mathsf{neg\_neg} (\mathsf{V}~\mathsf{c}_{e_1} \cdots \mathsf{c}_{e_n})$ is a conjunct of $\mathsf{E}_{e,\mathsf{V}}$.
Since $\mwrs{\mathsf{E}_{e,\mathsf{V}}}{\mathcal{M}}{s} \in  \{\T, \NF\}$,
we have that $\mwrs{\mathsf{neg\_neg} (\mathsf{V}~\mathsf{c}_{e_1} \cdots \mathsf{c}_{e_n})}{\mathcal{M}}{s} \in  \{\T, \NF\}$.
Then, we have $s(\mathsf{V})~e_1~\cdots~e_n = \mwrs{\mathsf{V}~\mathsf{c}_{e_1} \cdots \mathsf{c}_{e_n}}{\mathcal{M}}{s} \in  \{\T, \NF\}$.
If $e~e_1~\cdots~e_n = \F$, then $\mathsf{neg} (\mathsf{V}~\mathsf{c}_{e_1} \cdots \mathsf{c}_{e_n})$ is a conjunct of $\mathsf{E}_{e,\mathsf{V}}$.
Since $\mwrs{\mathsf{E}_{e,\mathsf{V}}}{\mathcal{M}}{s} \in  \{\T, \NF\}$, we have that $\mwrs{\mathsf{neg} (\mathsf{V}~\mathsf{c}_{e_1} \cdots \mathsf{c}_{e_n})}{\mathcal{M}}{s} \in  \{\T, \NF\}$.
Then, we have $s(\mathsf{V})~e_1~\cdots~e_n = \mwrs{\mathsf{V}~\mathsf{c}_{e_1} \cdots \mathsf{c}_{e_n}}{\mathcal{M}}{s} = \F$.
\end{proof}

Now, we are ready to define the program that constructs the constants $\ce{e}$ for all total elements $e$.

For any type $\pi$ and each total $e \in \mo{\pi}$ we define a program $\mathsf{P}_{\ce{e}}$ consisting of the following clauses:
For each $d_1,\ldots,d_n$ such that $e~d_1~\cdots~d_n\in\{\T,\NF\}$
\[
\begin{array}{rl}
  \ce{e}~\mathsf{R}_1~\cdots~\mathsf{R}_n \leftarrow & \Ee{\tv{d_1},\mathsf{R}_1}\wedge\cdots\wedge\Ee{\tv{d_n},\mathsf{R}_n} \wedge \\
  & \bigwedge\{ \cst{e}~(\mathsf{R}_i~\ce{e_{i,1}}~\cdots~\ce{e_{i,m}})~\mathsf{R}_1~\cdots~\mathsf{R}_n \mid d_i~e_{i,1}~\cdots~e_{i,m} \in \{\T,\NF\}\}
\end{array}
\]
$\mathsf{P}_{\tau}$ is the union of $\mathsf{P}^*_{\tau}$ and every $\mathsf{P}_{\ce{e}}$ for all total elements $e \in \mo{\pi}$ for every type $\pi$ with $\order(\pi) \leq \order(\tau)$.

Same as before, we will show that $\mathcal{M}$ is a model of $\mathsf{P}_{\tau}$.
Then, we will show that $\mathcal{M}$ is $\preceq$-minimal.
Since it is also total, $\mathcal{M}$ is an equilibrium model of $\mathsf{P}_{\tau}$.
Then, uniqueness follows from the fact that $\mathsf{P}_{\tau}$ is stratified.

\begin{lemma}
$\mathcal{M}$ is a model of $\mathsf{P}_{\tau}$.
\end{lemma}
\begin{proof}
Let $s$ be a state.
Let $C = \mathsf{c}_e~\mathsf{R}_1~\cdots~\mathsf{R}_n \leftarrow \mathsf{E}_{\tv{d_1},\mathsf{R}_1}\wedge\cdots\wedge\mathsf{E}_{\tv{d_n},\mathsf{R}_n} \wedge
\bigwedge\{ \cst{e}~(\mathsf{R}_i~\mathsf{c}_{e_{i,1}}~\cdots~\mathsf{c}_{e_{i,m}})~\mathsf{R}_1~\cdots~\mathsf{R}_n \mid d_i~e_{i,1}~\cdots~e_{i,m} \in \{\T,\NF\}\}$
be a clause in $\mathsf{P}_{\tau}$ for some $e \in \mo{\pi}$ for some type $\pi$ and $d_1,\ldots,d_n\in \arg(\pi)$ such that $e~d_1~\cdots~d_n\in\{\T,\NF\}$.
If $\mwrs{\mathsf{E}_{\tv{d_i},\mathsf{R}_i}}{\mathcal{M}}{s} = \F$ for some $i \in \{1,\ldots,n\}$, then we are done.
So, suppose that $\mwrs{\mathsf{E}_{\tv{d_i},\mathsf{R}_i}}{\mathcal{M}}{s} \in \{\T, \NF\}$ for all $i \in \{1,\ldots,n\}$.
By Lemma~\ref{lemma_E_2}, we have that $\tv{s(\mathsf{R}_i)} = \tv{d_i}$ for all $i \in \{1,\ldots,n\}$.
By the monotonicity of $e$ and Lemma~\ref{preceq_tv}, we have that $\tv{\mwrs{\mathsf{c}_e~\mathsf{R}_1~\cdots~\mathsf{R}_n}{\mathcal{M}}{s}} =
\tv{e~s(\mathsf{R}_1)~\cdots~s(\mathsf{R}_n)} = \tv{e~\tv{d_1}~\cdots~\tv{d_n}}$, which is $\T$ since $e~d_1~\cdots~d_n\in\{\T,\NF\}$.
Therefore, $\mwrs{\mathsf{c}_e~\mathsf{R}_1~\cdots~\mathsf{R}_n}{\mathcal{M}}{s} = e~s(\mathsf{R}_1)~\cdots~s(\mathsf{R}_n)$ is either $\T$ or $\NF$.
If it is $\T$, then we are done.
If it is $\NF$, since $e$ is total, there is some $i$ such that $s(\mathsf{R}_i)$ is not total.
In particular, there exist some total $e_{i,1},\ldots,e_{i,m}$ such that $s(\mathsf{R}_i)~e_{i,1}~\cdots~e_{i,m} = \NF$.
Then $\mwrs{\cst{e}~(\mathsf{R}_i~\mathsf{c}_{e_{i,1}}~\cdots~\mathsf{c}_{e_{i,m}})~\mathsf{R}_1~\cdots~\mathsf{R}_n}{\mathcal{M}}{s} =
\mathcal{M}^*(\cst{e})~(s(\mathsf{R}_i)~e_{i,1}~\cdots~e_{i,m})~s(\mathsf{R}_1)~\cdots~s(\mathsf{R}_n) =
\mathcal{M}^*(\cst{e})~\NF~s(\mathsf{R}_1)~\cdots~s(\mathsf{R}_n) =
e~s(\mathsf{R}_1)~\cdots~s(\mathsf{R}_n) = \NF$.
Since $\cst{e}~(\mathsf{R}_i~\mathsf{c}_{e_{i,1}}~\cdots~\mathsf{c}_{e_{i,m}})~\mathsf{R}_1~\cdots~\mathsf{R}_n$ is a conjunct of the body of $C$,
$\mathcal{M}$ satisfy $C$.
\end{proof}

\begin{lemma}
$\mathcal{M}$ is a $\preceq$-minimal model of $\mathsf{P}_{\tau}$.
\end{lemma}
\begin{proof}
Suppose, for the sake of contradiction, that there exists a model $\mathcal{N}$ of $\mathsf{P}_{\tau}$ such that $\mathcal{N} \prec \mathcal{M}$.
We can assume that $\mathcal{N}$ is $\preceq$-minimal.
Notice that the set $U$ of all constants $\cst{d}$ is a splitting set of $\mathsf{P}_{\tau}$ and $b_{U}(\mathsf{P_{\tau}}) = \mathsf{P}^*_{\tau}$.
By Lemma~\ref{splitting_submodel}, we have that $\mathcal{M}(\cst{d}) = \mathcal{N}(\cst{d})$ for all $d$.
Then, there exist some total $e \in \mo{\pi}$ for some type $\pi$ and some $(d_1,\ldots,d_n)\in arg(\pi)$ such that
$\mathcal{M}(\mathsf{c}_e)~d_1~\cdots~d_n=\T$ and $\mathcal{N}(\mathsf{c}_e)~d_1~\cdots~d_n=\NF$.
We assume that $\pi$ is of minimal \order, that is $\mathcal{N}(\ce{e'}) = \mathcal{M}(\ce{e'})$ for all $e'\in\pi'$ such that $\pi'$ has lower \order than $\pi$.
By the definition of $\mathcal{M}$, we have that $e~d_1~\cdots~d_n = \T$.
Then there exist in $\mathsf{P}_{\tau}$ the clause $C = \mathsf{c}_e~\mathsf{R}_1~\cdots~\mathsf{R}_n \leftarrow \mathsf{E}_{\tv{d_1},\mathsf{R}_1}\wedge\cdots\wedge\mathsf{E}_{\tv{d_n},\mathsf{R}_n} \wedge
\bigwedge\{ \cst{e}~(\mathsf{R}_i~\mathsf{c}_{e_{i,1}}~\cdots~\mathsf{c}_{e_{i,m}})~\mathsf{R}_1~\cdots~\mathsf{R}_n \mid d_i~e_{i,1}~\cdots~e_{i,m} \in \{\T,\NF\}\}$.
Let $s, s'$ be states such that $s(\mathsf{R}_i) = d_i$ and $s'(\mathsf{R}_i) = \tv{d_i}$ for all $i \in \{1,\ldots,n\}$.
By Lemma~\ref{lemma_E}, we have that $\mwrs{\mathsf{E}_{\tv{d_i},\mathsf{R}_i}}{\mathcal{N}}{s'} = \T$ for all $i \in \{1,\ldots,n\}$.
Notice that $s \preceq s'$. Using $\preceq$-monotonicity, we have that $\mwrs{\mathsf{E}_{\tv{d_i},\mathsf{R}_i}}{\mathcal{N}}{s}$ is either $\T$ or $\NF$.
But, by the definition of the expression $\mathsf{E}_{\tv{d_i},\mathsf{R}_i}$, all of its conjucts are based on $\mathsf{neg}$ or $\mathsf{neg\_neg}$.
For these predicates, $\mathcal{N}$ agrees with $\mathcal{M}$. And they cannot get the value $\NF$ in $\mathcal{M}$.
Therefore, $\mwrs{\mathsf{E}_{\tv{d_i},\mathsf{R}_i}}{\mathcal{N}}{s} = \T$ for all $i \in \{1,\ldots,n\}$.
Also, it is easy to see that $\mwrs{\cst{e}~(\mathsf{R}_i~\mathsf{c}_{e_{i,1}}~\cdots~\mathsf{c}_{e_{i,m}})~\mathsf{R}_1~\cdots~\mathsf{R}_n}{\mathcal{N}}{s} = \T$ for all $i \in \{1,\ldots,n\}$.
But, we have $\mwrs{\mathsf{c}_e~\mathsf{R}_1~\cdots~\mathsf{R}_n}{\mathcal{N}}{s} = \mathcal{N}(\mathsf{c}_e)~d_1~\cdots~d_n = \NF$,
which contradicts the fact that $\mathcal{N}$ is a model of $\mathsf{P}_{\tau}$.
\end{proof}

\section{Proofs of Section~\ref{sec:strong-eq}}\label{se-appendix}
\begin{retheorem}{se-theorem}
For any programs $\mathsf{P}_1$ and $\mathsf{P}_2$ the following two statements are equivalent:
\begin{enumerate}
\item For every program $\mathsf{P}$, $\mathsf{P}_1\cup \mathsf{P}$ and $\mathsf{P}_2\cup \mathsf{P}$
      have the same equilibrium models.
\item $\mathsf{P}_1$ and $\mathsf{P}_2$ have the same models.
\end{enumerate}
\end{retheorem}
\begin{proof}
To see that (2) implies (1), observe that $\mathsf{P}_1$ and $\mathsf{P}_2$ have the same models
and therefore $\mathsf{P}_1\cup \mathsf{P}$ and $\mathsf{P}_2\cup \mathsf{P}$ have the same models,
and consequently the same equilibrium models. We show that (1) implies (2).

Suppose that $\mathsf{P}_1$ and $\mathsf{P}_2$ do not have the same models. Without loss of generality, assume that $\mathsf{P}_1$ has a model ${\cal I}$ that is not a model of $\mathsf{P}_2$. Moreover, without loss of generality, assume that for every total relation $r$ there exists a predicate constant $\mathsf{c}_{r}$, different from all the constants that appear in $\mathsf{P}_1\cup \mathsf{P}_2$, and definable as implied by Theorem~\ref{definability_of_exact}. We assume that ${\cal I}$ assigns to $\mathsf{c}_{r}$ the value $r$. Notice that since $\mathsf{c}_{r}$ is definable by a stratified program, this program has a unique equilibrium model; therefore, for every ${\cal J}\prec {\cal I}$ that satisfies the defining clauses of $\mathsf{c}_{r}$, it holds ${\cal J}(\mathsf{c}_{r}) = {\cal I}(\mathsf{c}_{r}) = r$. We make corresponding assumptions about
non-total relations $r$ and predicate constants $\cst{r}$.

\vspace{0.2cm}
\noindent
\underline{\emph {Case 1:}} ${\cal I}$ is total.  We show how to find a program $\mathsf{P}$ such that ${\cal I}$ is an equilibrium model of  $\mathsf{P}_1\cup \mathsf{P}$ but not an equilibrium model of $\mathsf{P}_2\cup \mathsf{P}$. For every predicate constant $\mathsf{p}: \pi_1 \to \cdots \to \pi_n \to o$ in $\mathsf{P}_1 \cup \mathsf{P}_2$, we create a clause:
$$\mathsf{p}\,\mathsf{R}_1 \cdots \mathsf{R}_n \leftarrow \mbox{$\ce{{\cal I}(\mathsf{p})}\,\mathsf{R}_1 \cdots \mathsf{R}_n$.}$$
We take as $\mathsf{P}$ the union of all such clauses together with the defining clauses for the predicate constants $\ce{{\cal I}(\mathsf{p})}$, as constructed by the proof of Theorem~\ref{definability_of_exact}. It is easy to see that ${\cal I}$ is a model of $\mathsf{P}_1 \cup \mathsf{P}$: ${\cal I}$ is a model of $\mathsf{P}_1$ and also ${\cal I}$ trivially satisfies every rule of the above form. We claim that ${\cal I}$ is an equilibrium model of $\mathsf{P}_1 \cup \mathsf{P}$. Assume, for the sake of contradiction, that there exists some ${\cal J} \prec {\cal I}$ that is a model of $\mathsf{P}_1 \cup \mathsf{P}$. Then, there exists $\mathsf{p}$ and $d_1,\ldots,d_n$ such that ${\cal J}(\mathsf{p})\,d_1\cdots d_n = \NF$ and
${\cal I}(\mathsf{p})\,d_1\cdots d_n = \T$. But then, ${\cal J}$ does not satisfy the rule $\mathsf{p}\,\mathsf{R}_1 \cdots \mathsf{R}_n \leftarrow \mbox{$\ce{{\cal I}(\mathsf{p})}\,\mathsf{R}_1 \cdots \mathsf{R}_n$}$. This is a contradiction because we assumed that ${\cal J}$ is a model of $\mathsf{P}$. Thus, ${\cal I}$ is an equilibrium model of $\mathsf{P}_1 \cup \mathsf{P}$. However, it is not a model of $\mathsf{P}_2$, so it can not be a model of $\mathsf{P}_2\cup\mathsf{P}$.

\vspace{0.2cm}

\noindent
\underline{\emph{Case 2}}: ${\cal I}$ is not total.
We show how to find a program $\mathsf{P}$ such that there exists an interpretation ${\cal K}$
which is an equilibrium model of one of the programs $\mathsf{P}_1\cup \mathsf{P}$ and $\mathsf{P}_2\cup \mathsf{P}$ but not an equilibrium model of the other.
We define the interpretation ${\cal K}$ such that ${\cal K}(\mathsf{p}) = \tv{{\cal I}(\mathsf{p})}$ for any predicate constant $\mathsf{p}$ in $\mathsf{P}_1\cup\mathsf{P}_2$, and ${\cal K}(\mathsf{p})={\cal I}(\mathsf{p})$, otherwise.
Notice that ${\cal K}$ is total and ${\cal I} \preceq {\cal K}$. It is easy to verify that ${\cal K}$ is a model of $\mathsf{P}_1$.
If ${\cal K}$ is not a model of $\mathsf{P}_2$, we can follow the same arguments as in Case 1 for ${\cal K}$ instead of ${\cal I}$.
Therefore, we can assume that ${\cal K}$ is a model of $\mathsf{P}_2$.

Since ${\cal I}$ is not total, there exists a predicate constant $\mathsf{r}$ such that ${\cal I}(\mathsf{r})$ is not total.
Therefore, there exist total elements $e_1,\ldots,e_l$ such that ${\cal I}(\mathsf{r})\,e_1\cdots e_l = \NF$.
We denote by $\mathsf{E}_{U}$ the expression $(\mathsf{r}\,\ce{e_1}\cdots \ce{e_l})$.
We define $\mathsf{P}$ as the program consisting of the rules constructed with the following two ways:
\begin{itemize}
\item For every predicate constant $\mathsf{p}: \pi_1 \to \cdots \to \pi_n \to o$ in $\mathsf{P}_1\cup\mathsf{P}_2$, we create a clause:
      $$\mathsf{p}\,\mathsf{R}_1 \cdots \mathsf{R}_n \leftarrow \mbox{$\cst{{\cal I}(\mathsf{p})}\,\mathsf{E}_{U}~\mathsf{R}_1 \cdots \mathsf{R}_n$.}$$
\item Let $\mathsf{p}: \pi_{1}\to\cdots\to\pi_{n}\to o$ and $\mathsf{q}: \pi'_{1}\to\cdots\to\pi'_{m}\to o$ be predicate constants and let $a_1,\ldots,a_n$ and $b_1,\ldots,b_m$ be elements such that ${\cal I}(\mathsf{p})\,a_1 \cdots a_n = \NF$ and ${\cal I}(\mathsf{q})\,b_1 \cdots b_m = \NF$. Then, we construct the rule:
      $$\mathsf{q}\,\mathsf{R}_1 \cdots \mathsf{R}_m \leftarrow \mbox{$\mathsf{p}\,(\cst{a_1}~\mathsf{E}_{U})\cdots (\cst{a_n}~\mathsf{E}_{U}),(\ce{eq_{\pi'_{1}}}\,\mathsf{R}_1~(\cst{b_1}~\mathsf{E}_{U})), \ldots, (\ce{eq_{\pi'_{m}}}\,\mathsf{R}_m~(\cst{b_m}~\mathsf{E}_{U}))$.}$$
      where for any argument type $\pi$, $eq_{\pi}\in(\mo{\pi}\to \mo{\pi}\to \mo{o})$ is the relation such that for any $a,b \in \mo{\pi}$,
      \[
        eq_{\pi}\,a~b = \begin{cases}
          \T,  & \mbox{if $\tv{a}=\tv{b}$} \\
          \F, & \mbox{otherwise}
        \end{cases}
      \]
\end{itemize}
It is easy to see that $eq_{\pi}$ is a total $\preceq$-monotonic relation and therefore, by Theorem~\ref{definability_of_exact}, definable as a $\HOL$ program.

We claim that ${\cal K}$ satisfies every rule of $\mathsf{P}$. Assume we are given a rule of the form:
$$\mathsf{p}\,\mathsf{R}_1 \cdots \mathsf{R}_n \leftarrow \mbox{$\cst{{\cal I}(\mathsf{p})}\,\mathsf{E}_{U}~\mathsf{R}_1 \cdots \mathsf{R}_n$.}$$
and a state $s$ with $s(\mathsf{R}_i)=d_i$.
We have $\mwrs{\mathsf{E}_{U}}{{\cal K}}{s} = \T$.
It is easy to see that both the head and the body of the above rule evaluates to $\tv{{\cal I}(\mathsf{p})\,d_1 \cdots d_n}$ under
${\cal K}$ and $s$.
On the other hand, assume we are given a rule of the form:
      $$\mathsf{q}\,\mathsf{R}_1 \cdots \mathsf{R}_m \leftarrow \mbox{$\mathsf{p}\,(\cst{a_1}~\mathsf{E}_{U})\cdots (\cst{a_n}~\mathsf{E}_{U}),(\ce{eq_{\pi'_{1}}}\,\mathsf{R}_1~(\cst{b_1}~\mathsf{E}_{U})), \ldots, (\ce{eq_{\pi'_{m}}}\,\mathsf{R}_m~(\cst{b_m}~\mathsf{E}_{U}))$.}$$
If the body of this rule evaluates under ${\cal K}$ and $s$ to $\T$ or $\NF$, then $\tv{s(\mathsf{R}_i)}=\tv{b_i}$ for any $i\in\{1,\ldots,m\}$.
In that case, ${\cal K}(\mathsf{q})\,s(\mathsf{R}_1)\cdots s(\mathsf{R}_m) = \tv{{\cal I}(\mathsf{q})\,s(\mathsf{R}_1)\cdots s(\mathsf{R}_m)} = \tv{{\cal I}(\mathsf{q})\,b_1\cdots b_m} = \T$.

Consequently, ${\cal K}$ is a model of $\mathsf{P}_2 \cup \mathsf{P}$. Actually, it is an equilibrium model of $\mathsf{P}_2 \cup \mathsf{P}$; to see this, consider any model ${\cal J}$ of $\mathsf{P}_2 \cup \mathsf{P}$ such that ${\cal J} \prec {\cal K}$.
Let $\mathsf{p}$ be a predicate constant, $d_1,\ldots,d_n$ be some elements such that ${\cal I}(\mathsf{p})\,d_1\cdots d_n = \T$ and $s$ be a state with $s(\mathsf{R}_i)=d_i$.
Notice that $\mwrs{\cst{{\cal I}(\mathsf{p})}\,\mathsf{E}_{U}~\mathsf{R}_1 \cdots \mathsf{R}_n}{{\cal J}}{s} = \T$.
Since ${\cal J}$ is a model of $\mathsf{P}$, we must have $\mwrs{\mathsf{p}\,\mathsf{R}_1 \cdots \mathsf{R}_n}{{\cal J}}{s} = \T$, ie., ${\cal J}(\mathsf{p})\,d_1\cdots d_n = \T$.
This latter fact combined with the facts ${\cal I} \preceq {\cal K}$ and ${\cal J} \prec {\cal K}$ gives ${\cal I} \preceq {\cal J}$.
But it can not be the case that ${\cal I} = {\cal J}$, because, by assumption, ${\cal I}$ is not a model of $\mathsf{P}_2$ (while ${\cal J}$ is).
Therefore, ${\cal I} \prec {\cal J} \prec {\cal K}$.
Take a predicate constant $\mathsf{p}$ and elements $a_1,\ldots,a_n$ such that ${\cal I}(\mathsf{p})\,a_1\cdots a_n=\NF$ and ${\cal J}(\mathsf{p})\,a_1\cdots a_n=\T$ and a predicate constant $\mathsf{q}$ and elements $b_1,\ldots,b_m$ such that ${\cal J}(\mathsf{q})\,b_1\cdots b_m=\NF$ and ${\cal K}(\mathsf{q})\,b_1\cdots b_m=\T$.
Let $s$ be a state with $s(\mathsf{R}_i)=b_i$ for any $i\in\{1,\ldots,m\}$. Consider the rule:
      $$\mathsf{q}\,\mathsf{R}_1 \cdots \mathsf{R}_m \leftarrow \mbox{$\mathsf{p}\,(\cst{a_1}~\mathsf{E}_{U})\cdots (\cst{a_n}~\mathsf{E}_{U}),(\ce{eq_{\pi'_{1}}}\,\mathsf{R}_1~(\cst{b_1}~\mathsf{E}_{U})), \ldots, (\ce{eq_{\pi'_{m}}}\,\mathsf{R}_m~(\cst{b_m}~\mathsf{E}_{U}))$.}$$
But ${\cal J}$ does not satisfy this implication, contrary to the assumption that it is a model of $\mathsf{P}_2\cup \mathsf{P}$.

We show now that ${\cal K}$ is not an equilibrium model of $\mathsf{P}_1 \cup \mathsf{P}$. Consider the model
${\cal I}$ of $\mathsf{P}_1$. It is easy to see that ${\cal I}$ satisfies all elements of $\mathsf{P}$. Therefore, ${\cal I}$ is a model
of $\mathsf{P}_1 \cup \mathsf{P}$. But ${\cal I}$ is different than ${\cal K}$, because ${\cal K}$ is a model of
$\mathsf{P}_2$ while ${\cal I}$ is not. By definition of ${\cal K}$, ${\cal I} \prec {\cal K}$; therefore ${\cal K}$
is not an equilibrium model of $\mathsf{P}_1 \cup \mathsf{P}$.
\end{proof}

\section{Proofs of Section~\ref{sec:properties}}
\newcommand{\mtrue}{\mathit{true}}
\newcommand{\mfalse}{\mathit{false}}
\newcommand{\mwrstwo}[3]{\mathfrak{B}\lsem#1\rsem_{#3}(#2)}
\newcommand{\mwrc}[3]{\mathfrak{C}\lsem#1\rsem_{#3}(#2)}
\newcommand{\lfp}{\mathrm{lfp}}
\newcommand{\TP}{T_{\mathsf{P}}}
\newcommand{\AP}{A_{\mathsf{P}}}
\newcommand{\ATP}{\AP}
\newcommand{\motwo}[1]{\mathfrak{B}\mo{#1}}

In this section we give detailed proofs for the correspondence
between equilibrium models and stable models defined by~\cite{iclp24}.
We also include the definition of stable models for completeness.


\begin{definition}\label{def:orders_two-valued}
We define the two-valued meaning $\motwo{\pi}$ of a
type $\pi$, as follows:
\begin{itemize}
  \item $\motwo{\bool} = \{\mathit{true}, \mathit{false}\}$.
  The partial order $\leq_\bool$ is the one induced by $\mfalse <_\bool \mtrue$.
  \item $\motwo{\pi_1 \to \pi_2} = \motwo{\pi_1} \to \motwo{\pi_2}$.
        The partial order $\leq_{\pi_1 \to \pi_2}$ is defined as follows:
        for all $f,g \in \motwo{\pi_1 \to \pi_2}$,
        $f \leq_{\pi_1 \to \pi_2} g$ iff $f(d) \leq_{\pi_2} g(d)$ for all $d \in \motwo{\pi_1}$.
\end{itemize}
\end{definition}
For every predicate type $\pi$, $(\motwo{\pi}, \leq_\pi)$ is a complete lattice.
In the following, we denote by $\bigvee_{\leq_{\pi}}$ and
$\bigwedge_{\leq_{\pi}}$ the corresponding lub and glb operations of the above
lattice.

We also define a projection and expansion function that map
elements of the two-valued domain $\motwo{\pi}$ to the three-valued domain $\mo{\pi}$
and  vice-versa for each type.

\begin{definition}
Let $d \in \mo{\pi}$ and $e \in \motwo{\pi}$.
The projection $\pi(d) \in \motwo{\pi}$ and the
expansion $\epsilon(e) \in \mo{\pi}$ are defined recursively as follows.
\begin{itemize}
\item if $d \in \mo{o}$, $\pi(d) = \mfalse$ if $d = \F$, $\pi(d) = \mtrue$ otherwise and $\epsilon(\mfalse) = \F$, $\epsilon(\mtrue) = \T$.
\item if $d \in \mo{\pi \to \pi}$, $\pi(d) = \lambda x. \pi(d(\epsilon(x)))$ and $\epsilon(e) = \lambda x. \epsilon(e(\pi(x)))$.
\end{itemize}
\end{definition}

The following lemma establishes the relation between $\pi$ and $\epsilon$.
\begin{lemma}\label{pi_epsilon}\label{epsilon_pi}
Let $d \in \mo{\pi}$ and $e \in \motwo{\pi}$. Then,
\begin{itemize}
\item $\pi(\epsilon(e)) = e$.
\item $d \preceq \epsilon(\pi(d))$.
\end{itemize}
\end{lemma}
\begin{proof}
We will use induction on the structure of $\pi$.
If $\pi = o$ then the statement is trivial.
Now suppose that $\pi = \pi_1 \to \pi_2$ and the statement holds for $\pi_1$ and $\pi_2$.
Then, we have
\[
\begin{array}{rcll}
  \pi(\epsilon(e)) & = & \lambda x. \pi(\epsilon(e)(\epsilon(x)))  & \mbox{definition of $\pi$} \\
  & = & \lambda x. \pi((\lambda y. \epsilon(e(\pi(y))))(\epsilon(x)))  &\mbox{definition of $\epsilon$} \\
  & = & \lambda x. \pi(\epsilon(e(\pi(\epsilon(x))))) & \mbox{$\beta$-reduction} \\
  & = & \lambda x. \pi(\epsilon(e(x))) & \mbox{induction hypothesis for $\pi_1$} \\
  & = & \lambda x. e(x) & \mbox{induction hypothesis for $\pi_2$} \\
  & = & e &
\end{array}
\]

The second claim is also by induction on the structure of $\pi$.
If $\pi = o$ then the statement is trivial.
Now suppose that $\pi = \pi_1 \to \pi_2$ and the statement holds for $\pi_1$ and $\pi_2$.
Then, we have
\[
\begin{array}{rcll}
\epsilon(\pi(d)) & = & \lambda x. \epsilon(\pi(d)(\pi(x))) & \mbox{definition of $\epsilon$} \\
& = & \lambda x. \epsilon((\lambda y. \pi(d(\epsilon(y))))(\pi(x))) & \mbox{definition of $\pi$} \\
& = & \lambda x. \epsilon(\pi(d(\epsilon(\pi(x))))) & \mbox{$\beta$-reduction} \\
& \succeq & \lambda x. d(\epsilon(\pi(x))) & \mbox{induction hypothesis for $\pi_1$} \\
& \succeq & \lambda x. d(x) & \mbox{induction hypothesis for $\pi_2$} \\
& = & d
\end{array}
\]
This concludes the proof.
\end{proof}

We now characterize the behaviors of $\pi$ and $\epsilon$
with respect to the defined orderings.

\begin{lemma}\label{pi_relations}
Let $x, y \in \mo{\pi}$.
\begin{itemize}
\item if $x \leq y$ then $\pi(x) \leq \pi(y)$
\item if $x \preceq y$ then $\pi(x) = \pi(y)$
\end{itemize}
\end{lemma}
\begin{proof}
We will show the first statement. The proof of the second statement is similar.
We will use induction on the structure of $\pi$.
If $\pi = o$ then the statement is trivial.
Now suppose that $\pi = \pi_1 \to \pi_2$ and the statement holds for $\pi_1$ and $\pi_2$.
Let $x,y \in \mo{\pi_1 \to \pi_2}$ with $x \leq y$.
Let $e \in \motwo{\pi_1}$. By the definition of $\leq$, we have $x(\epsilon(e)) \leq y(\epsilon(e))$.
Using the induction hypothesis, we have $\pi(x(\epsilon(e))) \leq \pi(y(\epsilon(e)))$.
Since this holds for all $e \in \motwo{\pi_1}$, we have $\lambda e. \pi(x(\epsilon(e))) \leq \lambda e. \pi(y(\epsilon(e)))$.
We conclude that $\pi(x) \leq \pi(y)$.
\end{proof}

\begin{lemma}\label{epsilon_relations}
Let $x,y \in \motwo{\pi}$. If $x \leq y$ then $\epsilon(x) \leq \epsilon(y)$.
\end{lemma}
\begin{proof}
We will use induction on the structure of $\pi$.
If $\pi = o$ then the statement is trivial.
Now suppose that $\pi = \pi_1 \to \pi_2$ and the statement holds for $\pi_1$ and $\pi_2$.
Let $x,y \in \motwo{\pi_1 \to \pi_2}$ with $x \leq y$.
Let $d \in \mo{\pi_1}$. By the definition of $\leq$, we have $x(\pi(d)) \leq y(\pi(d))$.
Using the induction hypothesis, we have $\epsilon(x(\pi(d))) \leq \epsilon(y(\pi(d)))$.
Since this holds for all $d \in \mo{\pi_1}$, we have $\lambda d. \epsilon(x(\pi(d))) \leq \lambda d. \epsilon(y(\pi(d)))$.
We conclude that $\epsilon(x) \leq \epsilon(y)$.
\end{proof}

An interpretation $I$ of a program $\mathsf{P}$ assigns to
each predicate constant $\mathsf{p} : \pi$ of $\mathsf{P}$,
an element $I(\mathsf{p}) \in \motwo{\pi}$.
We will denote the set of the interpretations of a program $\mathsf{P}$
with $H_\mathsf{P}$. We define a partial order on $H_\mathsf{P}$ as
follows: for all $I, J \in H_\mathsf{P}$, $I \leq J$
iff for every predicate constant $\mathsf{p} : \pi$ that appears in
$\mathsf{P}$, $I(\mathsf{p}) \leq_{\pi} J(\mathsf{p})$.
It can be shown that
$(H_{\mathsf{P}}, \leq)$ is a complete lattice.
A state  $s$ of a program $\mathsf{P}$ is a function that
assigns to each variable $\mathsf{R}$ of type $\pi$, an element
$s(\mathsf{R}) \in \motwo{\pi}$. We denote the set of states with
$S_\mathsf{P}$.

\begin{definition}\label{standard-semantics}
Let $\mathsf{P}$ be an $\HOL$ program, $I$ an interpretation of $\mathsf{P}$, and
$s$ a state of $\mathsf{P}$. Then, the two-valued semantics of expressions is
defined as follows:
\begin{enumerate}
  \item $\mwrstwo{\mathsf{R}}{I}{s} = s(\mathsf{R})$
  \item $\mwrstwo{\mathsf{p}}{I}{s} = I(\mathsf{p})$
  \item $\mwrstwo{(\mathsf{E}_1\ \mathsf{E}_2)}{I}{s} = \mwrstwo{\mathsf{E}_1}{I}{s}\ \mwrstwo{\mathsf{E}_2}{I}{s}$
  \item $\mwrstwo{(\sim \mathsf{E})}{I}{s} = \sim \mwrstwo{\mathsf{E}}{I}{s}$
  where $\sim \mtrue = \mfalse$ and $\sim \mfalse = \mtrue$
\end{enumerate}
\end{definition}

A two-valued interpretation $M$ of $\mathsf{P}$
is a \emph{two-valued model} of $\mathsf{P}$ if and only if for every rule
$\mathsf{p}\ \mathsf{R}_1\ldots\mathsf{R}_n \lrule \mathsf{L}_1, \ldots, \mathsf{L}_m$ in $\mathsf{P}$
and for every state $s$,
$\bigwedge_{\leq_o}\{ \mwrstwo{\mathsf{L}_i}{M}{s} \mid i \in \{1,\ldots,m\} \} \leq_o \mwrstwo{\mathsf{p}\ \mathsf{R}_1\ldots\mathsf{R}_n}{M}{s}$.


The definitions of $\pi$ and $\epsilon$ can be extended naturally to interpretations and states.
If $I$ is an interpretation of $\mathsf{P}$
then $\pi(I)(\mathsf{p}) = \pi(I(\mathsf{p}))$ for every predicate symbol $\mathsf{p}$.
Similarly, if $s$ is a state of $\mathsf{P}$ then $\pi(s)(\mathsf{R}) = \pi(s(\mathsf{R}))$
for every variable $\mathsf{R}$. In the following lemma we establish the connection of the
semantic evaluation and the functions $\pi$ and $\epsilon$ for every expression $\mathsf{E}$.

\begin{lemma}\label{pi_expression}\label{epsilon_expression}
Let $\mathsf{E}$ be a $\HOL$ expression. Then,
\begin{itemize}
\item $\pi(\mwrs{\mathsf{E}}{\mathcal{I}}{s}) = \mwrstwo{\mathsf{E}}{\pi(\mathcal{I})}{\pi(s)}$.
\item $\epsilon(\mwrstwo{\mathsf{E}}{I}{s}) = \mwrs{\mathsf{E}}{\epsilon(I)}{\epsilon(s)}$.
\end{itemize}
\end{lemma}
\begin{proof}
The first statement is proved by induction on the structure of $\mathsf{E}$.
The only interesting case is when $\mathsf{E} = (\mathsf{E}_1~\mathsf{E}_2)$.
Then,
\[
\begin{array}{rcll}
  \mwrstwo{\mathsf{E}_1~\mathsf{E}_2}{\pi(\mathcal{I})}{\pi(s)} & = & \mwrstwo{\mathsf{E}_1}{\pi(\mathcal{I})}{\pi(s)}(\mwrstwo{\mathsf{E}_2}{\pi(\mathcal{I})}{\pi(s)}) & \\
  & = & \pi(\mwrs{\mathsf{E}_1}{\mathcal{I}}{s})(\pi(\mwrs{\mathsf{E}_2}{\mathcal{I}}{s})) & \mbox{induction hypothesis} \\
  & = & \lambda x. \pi(\mwrs{\mathsf{E}_1}{\mathcal{I}}{s}(\epsilon(x)))(\pi(\mwrs{\mathsf{E}_2}{\mathcal{I}}{s})) & \mbox{definition of $\pi$} \\
  & = & \pi(\mwrs{\mathsf{E}_1}{\mathcal{I}}{s}(\epsilon(\pi(\mwrs{\mathsf{E}_2}{\mathcal{I}}{s})))) & \mbox{$\beta$-reduction} \\
\end{array}
\]
By Lemma~\ref{epsilon_pi}, we have $\mwrs{\mathsf{E}_2}{\mathcal{I}}{s} \preceq \epsilon(\pi(\mwrs{\mathsf{E}_2}{\mathcal{I}}{s}))$.
By $\preceq$-monotonicity, $ \mwrs{\mathsf{E}_1}{\mathcal{I}}{s}(\mwrs{\mathsf{E}_2}{\mathcal{I}}{s}) \preceq \mwrs{\mathsf{E}_1}{\mathcal{I}}{s}(\epsilon(\pi(\mwrs{\mathsf{E}_2}{\mathcal{I}}{s})))$.
By Lemma~\ref{pi_relations}, we have $\pi(\mwrs{\mathsf{E}_1}{\mathcal{I}}{s}(\mwrs{\mathsf{E}_2}{\mathcal{I}}{s})) = \pi(\mwrs{\mathsf{E}_1}{\mathcal{I}}{s}(\epsilon(\pi(\mwrs{\mathsf{E}_2}{\mathcal{I}}{s}))))$.
We conclude that $\pi(\mwrs{\mathsf{E}_1~\mathsf{E}_2}{\mathcal{I}}{s}) = \mwrstwo{\mathsf{E}_1~\mathsf{E}_2}{\pi(\mathcal{I})}{\pi(s)}$.

The second statement is also proved by induction on the structure of $\mathsf{E}$.
The only interesting case is when $\mathsf{E} = (\mathsf{E}_1~\mathsf{E}_2)$.
Then,
\[
  \begin{array}{rcll}
    \mwrs{\mathsf{E}_1~\mathsf{E}_2}{\epsilon(I)}{\epsilon(s)} & = & \mwrs{\mathsf{E}_1}{\epsilon(I)}{\epsilon(s)}(\mwrs{\mathsf{E}_2}{\epsilon(I)}{\epsilon(s)}) & \\
    & = & \epsilon(\mwrstwo{\mathsf{E}_1}{I}{s})(\epsilon(\mwrstwo{\mathsf{E}_2}{I}{s})) & \mbox{induction hypothesis} \\
    & = & \lambda x. \epsilon(\mwrstwo{\mathsf{E}_1}{I}{s}(\pi(x)))(\epsilon(\mwrstwo{\mathsf{E}_2}{I}{s})) & \mbox{definition of $\epsilon$} \\
    & = & \epsilon(\mwrstwo{\mathsf{E}_1}{I}{s}(\pi(\epsilon(\mwrstwo{\mathsf{E}_2}{I}{s})))) & \mbox{$\beta$-reduction} \\
    & = & \epsilon(\mwrstwo{\mathsf{E}_1}{I}{s}(\mwrstwo{\mathsf{E}_2}{I}{s})) & \mbox{Lemma~\ref{pi_epsilon}} \\
    & = & \epsilon(\mwrstwo{\mathsf{E}_1~\mathsf{E}_2}{I}{s}) &
  \end{array}
\]
This concludes the proof.
\end{proof}

\begin{lemma}\label{pi_model}
Let $\mathsf{P}$ be a $\HOL$ program.
If $\mathcal{M}$ is a model of $\mathsf{P}$ then $\pi(\mathcal{M})$ is a two-valued model of $\mathsf{P}$.
\end{lemma}
\begin{proof}
Let $\mathsf{H} \leftarrow \mathsf{B}$ be a clause of $\mathsf{P}$ and $s$ be a two valued state.
Since $\mathcal{M}$ is a model of $\mathsf{P}$, we have $\mwrs{\mathsf{B}}{\mathcal{M}}{\epsilon(s)} \leq \mwrs{\mathsf{H}}{\mathcal{M}}{\epsilon(s)}$.
By Lemma~\ref{pi_relations}, $\pi(\mwrs{\mathsf{B}}{\mathcal{M}}{\epsilon(s)}) \leq \pi(\mwrs{\mathsf{H}}{\mathcal{M}}{\epsilon(s)})$.
By Lemma~\ref{pi_expression}, $\mwrs{\mathsf{B}}{\pi(\mathcal{M})}{\pi(\epsilon(s))} \leq \mwrs{\mathsf{H}}{\pi(\mathcal{M})}{\pi(\epsilon(s))}$.
By Lemma~\ref{pi_epsilon}, $\pi(\epsilon(s)) = s$. Therefore, $\mwrs{\mathsf{B}}{\pi(\mathcal{M})}{s} \leq \mwrs{\mathsf{H}}{\pi(\mathcal{M})}{s}$.
Since this holds for all clauses and two-valued states, $\pi(\mathcal{M})$ is a two-valued model of $\mathsf{P}$.
\end{proof}

\begin{lemma}\label{epsilon_model}
Let $\mathsf{P}$ be a $\HOL$ program.
If $M$ is a two-valued model of $\mathsf{P}$ then $\epsilon(M)$ is a model of $\mathsf{P}$.
\end{lemma}
\begin{proof}
Let $\mathsf{p}~\mathsf{R}_1\cdots\mathsf{R}_n \leftarrow \mathsf{B}$ be a clause of $\mathsf{P}$ and $s$ be a state.
Since $M$ is a two-valued model of $\mathsf{P}$, we have $\mwrstwo{\mathsf{B}}{M}{\pi(s)} \leq \mwrstwo{\mathsf{p}~\mathsf{R}_1\cdots\mathsf{R}_n}{M}{\pi(s)}$.
By Lemma~\ref{epsilon_relations}, we have $\epsilon(\mwrstwo{\mathsf{B}}{M}{\pi(s)}) \leq \epsilon(\mwrstwo{\mathsf{p}~\mathsf{R}_1\cdots\mathsf{R}_n}{M}{\pi(s)})$.
By Lemma~\ref{epsilon_expression}, we have $\mwrs{\mathsf{B}}{\epsilon(M)}{\epsilon(\pi(s))} \leq \mwrs{\mathsf{p}~\mathsf{R}_1\cdots\mathsf{R}_n}{\epsilon(M)}{\epsilon(\pi(s))}$.
We take cases for $\mwrs{\mathsf{B}}{\epsilon(M)}{s}$.
If it is $\F$, we are done.
Suppose that $\mwrs{\mathsf{B}}{\epsilon(M)}{s} \in \{\T,\NF\}$.
By Lemma~\ref{epsilon_pi}, we have $s \preceq \epsilon(\pi(s))$.
Therefore, $\mwrs{\mathsf{B}}{\epsilon(M)}{s} \preceq \mwrs{\mathsf{B}}{\epsilon(M)}{\epsilon(\pi(s))}$,
so that $\mwrs{\mathsf{B}}{\epsilon(M)}{\epsilon(\pi(s))} \in \{\T,\NF\}$.
By the previous inequality, we have $\NF \leq \mwrs{\mathsf{p}~\mathsf{R}_1\cdots\mathsf{R}_n}{\epsilon(M)}{\epsilon(\pi(s))}$.
Since $s \preceq \epsilon(\pi(s))$, we have $\mwrs{\mathsf{p}~\mathsf{R}_1\cdots\mathsf{R}_n}{\epsilon(M)}{s} \in \{\T,\NF\}$.
But $\mwrs{\mathsf{p}~\mathsf{R}_1\cdots\mathsf{R}_n}{\epsilon(M)}{s} = \epsilon(M(\mathsf{p}))(s(\mathsf{R}_1)\cdots(s(\mathsf{R}_n)))$.
By the definition of $\epsilon$, this cannot take the value $\NF$, so that $\mwrs{\mathsf{p}~\mathsf{R}_1\cdots\mathsf{R}_n}{\epsilon(M)}{s} = \T$.
In any case $\mwrs{\mathsf{B}}{\epsilon(M)}{s} \leq \mwrs{\mathsf{p}~\mathsf{R}_1\cdots\mathsf{R}_n}{\epsilon(M)}{s}$.
Since this holds for all clauses and states, $\epsilon(M)$ is a model of $\mathsf{P}$.
\end{proof}

The following lemma shows that $\pi$ is injective on the set of
equilibrium models and thus distinct equilibrium models have distinct projections.

\begin{lemma}\label{pi_injective_on_equilibriums}
Let $\mathsf{P}$ be a program and $\mathcal{M}_1$, $\mathcal{M}_2$
be equilibrium models of $\mathsf{P}$.
If $\pi(\mathcal{M}_1) = \pi(\mathcal{M}_2)$ then $\mathcal{M}_1 = \mathcal{M}_2$.
\end{lemma}
\begin{proof}
 Since $\pi(\mathcal{M}_1) = \pi(\mathcal{M}_2)$,
$\mathcal{M}_1$ and $\mathcal{M}_2$ must agree on the $\F$ values,
i.e., for all predicates symbols
$\mathsf{p}$ and values $d_1, \ldots, d_n$, $\mathcal{M}_1(\mathsf{p})\ d_1 \cdots d_n = \F$
if and only if $\mathcal{M}_2(\mathsf{p})\ d_1 \cdots d_n = \F$.
Consider the interpretation $\mathcal{M}' = \mathcal{M}_1 \wedge \mathcal{M}_2$.
$\mathcal{M}'$ must also agree with $\mathcal{M}_1$ and $\mathcal{M}_2$
on the $\F$ values. Therefore, $\mathcal{M}' \preceq \mathcal{M}_1$ because
$\mathcal{M}' \leq \mathcal{M}_1$ and similarly, $\mathcal{M}' \preceq \mathcal{M}_2$.
We will show that $\mathcal{M}'$ is a model of $\mathsf{P}$ which is a
contradiction because $\mathcal{M}_1$ and $\mathcal{M}_2$ are equilibrium models
and therefore $\preceq$-minimal.

Consider a rule $\mathsf{p}\ \mathsf{R}_1\cdots\mathsf{R}_n \leftarrow \mathsf{B}$ in $\mathsf{P}$.
We distinguish cases on the value of $\mo{\mathsf{p}\ \mathsf{R}_1\cdots\mathsf{R}_n}_s(\mathcal{M}')$.
Assume that $\mo{\mathsf{p}\ \mathsf{R}_1\cdots\mathsf{R}_n}_s(\mathcal{M}') = \F$ for some state $s$.
Then, $\mo{\mathsf{p}\ \mathsf{R}_1\cdots\mathsf{R}_n}_s(\mathcal{M}_1) = \F$. But since $\mathcal{M}_1$
is a model of $\mathsf{P}$ it must be $\mo{\mathsf{B}}_s(\mathcal{M}_1) = \F$. Therefore,
$\mo{\mathsf{B}}_s(\mathcal{M}') = \F$ and $\mathsf{M}'$ satisfies the rule.
If $\mo{\mathsf{p}\ \mathsf{R}_1\cdots\mathsf{R}_n}_s(\mathcal{M}') = \T$ then the rule is also satisfied.
If $\mo{\mathsf{p}\ \mathsf{R}_1\cdots\mathsf{R}_n}_s(\mathcal{M}') = \NF$ either
$\mo{\mathsf{p}\ \mathsf{R}_1\cdots\mathsf{R}_n}_s(\mathcal{M}_1) = \NF$ or
$\mo{\mathsf{p}\ \mathsf{R}_1\cdots\mathsf{R}_n}_s(\mathcal{M}_2) = \NF$.
Assume, without loss of generality, that
$\mo{\mathsf{p}\ \mathsf{R}_1\cdots\mathsf{R}_n}_s(\mathcal{M}_1) = \NF$. Since
$\mathcal{M}_1$ is a model of $\mathsf{P}$ it must be $\mo{\mathsf{B}}_s(\mathcal{M}_1) \leq \NF$
which implies that $\mo{\mathsf{B}}_s(\mathcal{M}') \leq \NF$. Therefore, $\mathcal{M}'$ satisfies the rule.
\end{proof}

Finally, we will define the concept of stable models as fixpoints of an approximation operator.
First we will define an auxiliary semantic function that operates on two two-valued interpretations
which we call pair semantics and then define the approximation operator based on those semantics.

For the following definition we need the set $H^c_\mathsf{P}$ which is defined as
$H^c_\mathsf{P} = \{ (I, J) \mid I \in H_\mathsf{P}, J \in H_\mathsf{P}, I \leq J\}$.

\begin{definition}\label{def:pair-semantics}
Let $\mathsf{P}$ be a program, $(I, J) \in H^c_\mathsf{P}$, and $s \in S_{\mathsf{P}}$.
The pair semantics of expressions is defined as follows:
\begin{enumerate}
  \item $\mwrc{\mathsf{R}}{I,J}{s} = (s(\mathsf{R}),s(\mathsf{R}))$
  \item $\mwrc{\mathsf{p}}{I,J}{s} = (I(\mathsf{p}), J(\mathsf{p}))$
  \item $\mwrc{(\mathsf{E}_1\ \mathsf{E}_2)}{I,J}{s} = (\bigwedge_{\leq_\pi} \{f(d) \mid d \in \lsem \rho\rsem, l \leq d \leq u \}, \bigvee_{\leq_\pi} \{g(d) \mid d \in \lsem \rho\rsem, l \leq d \leq u\})$,
      where $(f,g) = \mwrc{\mathsf{E}_1}{I,J}{s}$, $(l,u) = \mwrc{\mathsf{E}_2}{I,J}{s}$ for $\mathsf{E}_1\! :\! \rho \to \pi$ and $\mathsf{E}_2\! :\! \rho$.
  \item $\mwrc{(\sim \mathsf{E})}{I,J}{s} =  \mwrc{\mathsf{E}}{I,J}{s}^{-1}$,
  with $(l,u)^{-1} = (\sim u, \sim l)$
\end{enumerate}
\end{definition}

\begin{definition}\label{def:pair-ap}
Let $\mathsf{P}$ be a program. The \emph{approximating operator}
$\ATP : H^c_\mathsf{P} \to H^c_\mathsf{P}$ is defined
for every predicate constant $\mathsf{p} : \rho_1 \to \cdots \to \rho_n \to \bool$ in $\mathsf{P}$ and
all $d_1 \in \mo{\rho_1},\ldots, d_n \in \mo{\rho_n}$, as $\ATP(I,J) = (\ATP(I,J)_1, \ATP(I,J)_2)$
where, for $i \in \{ 1, 2 \}$:
\[ \ATP(I,J)_i(\mathsf{p}) (\bar{d})= \bigvee\nolimits_{\leq_\bool}\{
  [\mwrc{\mathsf{B}}{I,J}{s[\bar{\mathsf{R}}/\bar{d}]}]_i \mid
          \mbox{$s\in S_{\mathsf{P}}$ and
                $(\mathsf{p}\ \bar{\mathsf{R}} \lrule \mathsf{B})$ in $\mathsf{P}$}\}
\]
\end{definition}

Let $\mathsf{P}$ be a $\HOL$ program and $I$ a
two-valued interpretation of $\mathsf{P}$.
We will say that $(I,J)$ is a stable fixpoint of $\ATP$ iff
$I = \lfp \ATP(\cdot,J)_1$ and $J = \lfp \ATP(I,\cdot)_2$.
A stable model of $\mathsf{P}$ is a model $I$ of $\mathsf{P}$ 
such that $(I,I)$ is a stable fixpoint of $\ATP$.

\begin{lemma}\label{splitting_aft_submodel}
Let $\mathsf{P}$ be a program and $U$ be a non-empty splitting set of $\mathsf{P}$.
If $\mathcal{M}$ is a stable model of $\mathsf{P}$,
then $\mathcal{M}$ restricted to $U$ is a stable model of $b_U(\mathsf{P})$.
\end{lemma}
\begin{proof}
Let $\ATP$ be the approximator of $\mathsf{P}$ and $I,J$ a pair interpretation for $\mathsf{P}$. For any predicate $\mathsf{p}$ in $\mathsf{P}$ we have
\[ \ATP(I,J)_i(\mathsf{p}) (\bar{d})= \bigvee\nolimits_{\leq_o}\{
  [\mwrc{\mathsf{B}}{I,J}{s[\bar{\mathsf{R}}/\bar{d}]}]_i \mid
          \mbox{$s\in S_{\mathsf{P}}$ and
                $(\mathsf{p}\ \bar{\mathsf{R}} \lrule \mathsf{B})$ in $\mathsf{P}$}\}
\]
Similarly, let $A_{b_U(\mathsf{P})}$ be the approximator of $b_U(\mathsf{P})$. For $I_U,J_U$ a pair interpretation for $\mathsf{P}$
it is defined for any predicate $\mathsf{p}$ in $b_U(\mathsf{P})$ as
\[ A_{b_U(\mathsf{P})}(I_U,J_U)_i(\mathsf{p}) (\bar{d})= \bigvee\nolimits_{\leq_o}\{
  [\mwrc{\mathsf{B}}{I_U,J_U}{s[\bar{\mathsf{R}}/\bar{d}]}]_i \mid
          \mbox{$s\in S_{b_U(\mathsf{P})}$ and
                $(\mathsf{p}\ \bar{\mathsf{R}} \lrule \mathsf{B})$ in ${b_U(\mathsf{P})}$}\}
\]

Let $(I,J)$ be a stable fixpoint for $\mathsf{P}$ and let $(I_U,J_U)$ be its restriction to the bottom program ${b_U(\mathsf{P})}$.
It is easy to verify that $(I_U,J_U)$ is a fixpoint of the approximator $A_{b_U(\mathsf{P})}$. But then $A_{b_U(\mathsf{P})}(X,J_U)_1$ and
$A_{b_U(\mathsf{P})}(I_U,X)_2$ are well-defined operators in the corresponding intervals. In example for any
$X \leq J_U$ we have that $A_{b_U(\mathsf{P})}(X,J_U)_1 \leq A_{b_U(\mathsf{P})}(J_U,J_U)_1
\leq A_{b_U(\mathsf{P})}(J_U,J_U)_2 \leq A_{b_U(\mathsf{P})}(I_U,J_U)_2 = J_U$.

For the sake of contradiction assume that $I_U' = \lfp A_{b_U(\mathsf{P})}(\cdot,J_U)_1 < I_U$. We choose an interpretation $I'$
such that
\begin{itemize}
      \item $I'(\mathsf{p})=I_U'(\mathsf{p})$ for $\mathsf{p}$ in ${b_U(\mathsf{P})}$.
      \item $I'(\mathsf{p})=I(\mathsf{p})$ otherwise.
\end{itemize}
It is easy to verify that $I' < I$.
We show that $I'$ is a prefixpoint of $\ATP(\cdot,J)_1$ which is a contradiction since $I$ is the least fixpoint.
Indeed, for any $\mathsf{p}$ in $b_U(\mathsf{P})$ we have that
\begin{align*}
 \ATP(I',J)_1(\mathsf{p}) (\bar{d})
  & = \bigvee\nolimits_{\leq_o}\{
      [\mwrc{\mathsf{B}}{I',J}{s[\bar{\mathsf{R}}/\bar{d}]}]_1 \mid \mbox{$s\in S_{\mathsf{P}}$ and
                                                                   $(\mathsf{p}\ \bar{\mathsf{R}} \lrule \mathsf{B})$ in $\mathsf{P}$}\}\\
  & = \bigvee\nolimits_{\leq_o}\{
      [\mwrc{\mathsf{B}}{I'_U,J_U}{s_u[\bar{\mathsf{R}}/\bar{d}]}]_1 \mid \mbox{$s_u\in S_{b_U(\mathsf{P})}$ and
                                                                   $(\mathsf{p}\ \bar{\mathsf{R}} \lrule \mathsf{B})$ in ${b_U(\mathsf{P})}$}\}\\
  & = A_{b_U(\mathsf{P})}(I'_U,J_U)_1(\mathsf{p}) (\bar{d})= I'_U(\mathsf{p}) (\bar{d}) = I'(\mathsf{p}) (\bar{d})
\end{align*}
where we used the fact that all rules with head $\mathsf p$ are in $b_U(\mathsf P)$,
their bodies mention only $U$-predicates, and on those $I'$ agrees with $I'_U$. 

For any $\mathsf{p}$ not in $U$ it follows by the monotonicity of $\ATP$ that 
$\ATP(I',J)_1(\mathsf{p}) \leq \ATP(I,J)_1(\mathsf{p})=I(\mathsf{p})$. Furthermore, 
since $I(\mathsf{p}) = I'(\mathsf{p})$
it follows $\ATP(I',J)_1(\mathsf{p}) \leq I'(\mathsf{p})$.

The case for $J_U$ is analogous.
\end{proof}

\begin{retheorem}{aft-coincide-equilibrium-on-stratified}
Let $\mathsf{P}$ be a stratified program.
Let ${\cal M}$ be its unique equilibrium model and ${\cal N}$ be its unique 2-valued stable model under the AFT semantics.
Then, $\pi({\cal M}) = {\cal N}$, where $\pi$ is a collapse function from ${\cal H}$
to the set of 2-valued interpretations.
%
\end{retheorem}
\begin{proof}
For every $i \in \{1, \ldots, k\}$, we are going to show that
$\pi({\cal M}_i)$ is the unique 2-valued stable model of $\mathsf{P}_1 \cup \cdots \cup \mathsf{P}_i$ under AFT
using induction.
For each $i \in \{1, \ldots, k\}$, we call ${\cal N}_i$ the unique 2-valued model of $\mathsf{P}_1 \cup \cdots \cup \mathsf{P}_i$ under AFT.
For the basis case, let $i=1$.
By Lemma~\ref{epsilon_model}, we have that $\epsilon({\cal N}_1)$ is a model of $\mathsf{P}$.
Also, by Lemma~\ref{minimality}, we have that ${\cal M}_1$ is the minimum model of $\mathsf{P}_1$.
Therefore, ${\cal M}_1\leq \epsilon({\cal N}_1)$.
By Lemma~\ref{pi_relations}, we have that $\pi({\cal M}_1) \leq \pi(\epsilon({\cal N}_1))$.
But, by Lemma~\ref{pi_epsilon}, we have that $\pi(\epsilon({\cal N}_1)) = {\cal N}_1$, so that $\pi({\cal M}_1) \leq {\cal N}_1$.
By Lemma~\ref{pi_model}, $\pi({\cal M}_1)$ is a 2-valued model of $\mathsf{P}_1$ and by \cite[Theorem~7.2]{iclp24} ${\cal N}_1$ is minimal.
Therefore, $\pi({\cal M}_1) = {\cal N}_1$.

Now suppose that $\pi({\cal M}_{i-1}) = {\cal N}_{i-1}$ for some i $\in \{2, \ldots, k\}$. We will show that $\pi({\cal M}_i) = {\cal N}_i$.
We define an interpretation ${\cal I}$ of $\mathsf{P}_1 \cup \cdots \cup \mathsf{P}_i$ as follows:

\[{\cal I}(\mathsf{p}) = \begin{cases}
  {\cal M}_{i-1}(\mathsf{p}) & \text{if ${\cal M}_{i-1}(\mathsf{p})$ is defined}\\
  \epsilon({\cal N}_i)(\mathsf{p}) & \text{otherwise}
\end{cases}\]

\medskip
\noindent
\underline{\emph {Claim 1:}} ${\cal I} \preceq \epsilon({\cal N}_i)$.
Let $\mathsf{p}$ be a predicate constant symbol defined in $\mathsf{P}_1 \cup \cdots \cup \mathsf{P}_{i}$.
If ${\cal M}_{i-1}(\mathsf{p})$ is not defined, then ${\cal I}(\mathsf{p}) = \epsilon({\cal N}_i)(\mathsf{p})$.
If ${\cal M}_{i-1}(\mathsf{p})$ is defined, then, using the induction hypothesis, we have that $\pi({\cal M}_{i-1}) = {\cal N}_{i-1}$.
Then, $\epsilon(\pi({\cal M}_{i-1})) = \epsilon({\cal N}_{i-1})$.
By Lemma~\ref{epsilon_pi}, ${\cal M}_{i-1}(\mathsf{p}) \preceq \epsilon(\pi({\cal M}_{i-1}))(\mathsf{p})$ = $\epsilon({\cal N}_{i-1})(\mathsf{p})$.
By Lemma~\ref{splitting_aft_submodel}, ${\cal N}_{i-1}(\mathsf{p}) = {\cal N}_i(\mathsf{p})$.
Therefore, ${\cal I}(\mathsf{p}) = {\cal M}_{i-1}(\mathsf{p}) \preceq \epsilon({\cal N}_i)(\mathsf{p})$.

\medskip
\noindent
\underline{\emph {Claim 2:}} ${\cal I}$ is a model of $\mathsf{P}_1 \cup \cdots \cup \mathsf{P}_i$.
Let $C = \mathsf{p}~\mathsf{R}_1~\cdots~\mathsf{R}_n \leftarrow \mathsf{B}$ be a clause of $\mathsf{P}_1 \cup \cdots \cup \mathsf{P}_i$.
If ${\cal M}_{i-1}(\mathsf{p})$ is defined, then $C$ is a clause of $\mathsf{P}_1 \cup \cdots \cup \mathsf{P}_{i-1}$.
Then for every predicate constant symbol $\mathsf{q}$ appearing in $C$, ${\cal M}_{i-1}(\mathsf{q})$ would also be defined.
Since ${\cal M}_{i-1}$ satisfies $C$, it follows that ${\cal I}$ satisfies $C$.
Suppose now that ${\cal M}_{i-1}(\mathsf{p})$ is not defined.
Let $s$ be a state and let $d_1,\ldots,d_n$ be elements such that $d_j = s(\mathsf{R}_j)$ for any $j \in \{1,\ldots,n\}$.
Then $\mwrs{\mathsf{p}~\mathsf{R}_1~\cdots~\mathsf{R}_n}{s}{\cal I} = \epsilon({\cal N}_i)(\mathsf{p})~d_1~\cdots~d_n$.
By the definition of $\epsilon$,  $\epsilon({\cal N}_i)(\mathsf{p})~d_1~\cdots~d_n$ can't be $\NF$.
If it is $\T$, then ${\cal I}$ satisfies $C$.
So, suppose that it is $\F$.
By Lemma~\ref{epsilon_model}, $\epsilon({\cal N}_i)$ is a model of $\mathsf{P}_1 \cup \cdots \cup \mathsf{P}_i$,
so that $\epsilon({\cal N}_i)$ satisfies $C$.
Since $\epsilon({\cal N}_i)(\mathsf{p})~d_1~\cdots~d_n = \F$, $\mwrs{\mathsf{B}}{s}{\epsilon({\cal N}_i)}$ has to be $\F$.
By Claim~1 and $\preceq$-monotonicity, $\mwrs{\mathsf{B}}{s}{\cal I} = \F$.
Therefore, in any case, ${\cal I}$ satisfies $C$.

\medskip
\noindent
\underline{\emph {Claim 3:}} $\pi({\cal I}) = {\cal N}_i$.
Let $\mathsf{p}$ be a predicate constant symbol.
If ${\cal M}_{i-1}(\mathsf{p})$ is defined, then $\pi({\cal I}(\mathsf{p})) = \pi({\cal M}_{i-1}(\mathsf{p}))$.
By the induction hypothesis, $\pi({\cal M}_{i-1}(\mathsf{p})) = {\cal N}_{i-1}(\mathsf{p})$.
By Lemma~\ref{splitting_aft_submodel}, ${\cal N}_{i-1}(\mathsf{p}) = {\cal N}_i(\mathsf{p})$.
If ${\cal M}_{i-1}(\mathsf{p})$ is not defined, then $\pi({\cal I}(\mathsf{p})) = \pi(\epsilon({\cal N}_i)(\mathsf{p}))$.
By Lemma~\ref{pi_epsilon}, $\pi(\epsilon({\cal N}_i)(\mathsf{p})) = {\cal N}_i(\mathsf{p})$.

It is easy to see that ${\cal I}$ agrees with ${\cal M}_{i-1}$.
By Lemma~\ref{minimality}, we have that ${\cal M}_i$ is the minimum model of $\mathsf{P}_1 \cup \cdots \cup \mathsf{P}_i$ that agree with ${\cal M}_{i-1}$.
Therefore ${\cal M}_i \leq {\cal I}$.
By Lemma~\ref{pi_relations}, we have that $\pi({\cal M}_i) \leq \pi({\cal I})$.
By Claim~3, $\pi({\cal I}) = {\cal N}_i$, so that $\pi({\cal M}_i) \leq {\cal N}_i$.
By Lemma~\ref{pi_model}, $\pi({\cal M}_i)$ is a 2-valued model of $\mathsf{P}_1 \cup \cdots \cup \mathsf{P}_i$ and by \cite[Theorem~7.2]{iclp24} ${\cal N}_i$ is minimal.
Therefore, $\pi({\cal M}_i) = {\cal N}_i$.
\end{proof}